\documentclass[nofootinbib,reprint,amssymb,superscriptaddress,twocolumn]{revtex4-2}

\usepackage[pdftex]{graphicx}
\usepackage{subfig}
\usepackage{tabularx}
\usepackage{xcolor}
\usepackage{tikz}
\usepackage{amsmath,amsfonts,amssymb,mathtools}
\usepackage{amsthm}

\newtheorem{theorem}{Theorem}       
\newtheorem{lemma}[theorem]{Lemma}
\newtheorem{proposition}[theorem]{Proposition}
\newtheorem{corollary}[theorem]{Corollary} 
\newtheorem{definition}{Definition}

\usepackage{tasks}
\usepackage{csquotes}
\usepackage{diagbox}
\usepackage{float}
\usepackage[utf8]{inputenc}
\usepackage{csquotes}

\usepackage{float}

\usepackage{textcomp}
\usepackage{bbm,bm}
\usepackage[normalem]{ulem}
\usepackage{comment}
\usepackage{physics}
\usepackage[colorlinks]{hyperref}

\newcommand{\R}{\mathbb R}
\newcommand{\Q}{\mathbb Q}
\newcommand{\C}{\mathbb C}
\newcommand{\Herm}{\operatorname{Herm}}
\newcommand{\Sym}{\operatorname{Sym}}
\newcommand{\Sep}{\operatorname{Sep}}
\newcommand{\dist}{\operatorname{dist}}
\newcommand{\one}{\mathbf 1}

\usepackage{pgfplots}
\definecolor{blueprl}{RGB}{46,48,146}
\usepgfplotslibrary{groupplots}
\usepackage{lipsum}

\newcommand{\N}{\mathbb{N}}
\newcommand{\Z}{\mathbb{Z}}

\newcommand{\sans}[1]{\text{\normalfont\fontfamily{lmss}\selectfont #1}}
\newcommand{\np}{\sans{NP}}
\newcommand{\p}{\sans{P}}

\newcommand{\orth}{\mathrm{O}}
\newcommand{\wt}{\widetilde}
\newcommand{\supp}{\operatorname{supp}}
\numberwithin{equation}{section}
\newcommand{\supplementtocstart}{}

\def\QUICS{Joint Center for Quantum Information and Computer Science,
NIST/University of Maryland, College Park, Maryland 20742, USA}
\def\JQI{Joint Quantum Institute, NIST/University of Maryland, College Park, Maryland 20742, USA}
\def\ASTAR{Quantum Innovation Centre (Q.InC), Agency for Science, Technology and Research (A*STAR), 4 Fusionopolis Way, Kinesis \#05-01, Singapore 138635, Republic of Singapore}
\def\CQT{Centre for Quantum Technologies (CQT), National University of Singapore, Singapore 117543, Republic of Singapore}
\def\UMIACS{UMIACS, University of Maryland, College Park, Maryland 20742, USA}
\def\UEC{Graduate School of Informatics and Engineering, The
  University of Electro-Communications, Tokyo 182-8585, Japan}
\begin{document}
 \title{Deciding the Attainability of the Multiparameter\\ Quantum Fisher Information is \np-Hard}

\author{Lorc{\'a}n O. Conlon}
\email{lorcanconlon@gmail.com}
\affiliation{\JQI}
\affiliation{\QUICS}
\author{Laura Shou}
\affiliation{\JQI}
\affiliation{\QUICS}
\author{V Vijendran}
\affiliation{\CQT}
\affiliation{\ASTAR}
\author{Kishor Bharti}
\altaffiliation{Currently at IonQ Inc.}
\affiliation{\QUICS}
\affiliation{\UMIACS}
\author{Aritra Das}
\affiliation{ Center for Quantum Software and Innovation,
 University of Technology Sydney, Ultimo, NSW 2007, Australia}
 \author{Simon K. Yung}
\affiliation{Department of Quantum Science and Technology, Research School of Physics, The Australian National University, Canberra, ACT 2601, Australia}
\author{Jacob M. Taylor}
\affiliation{\QUICS}
\affiliation{Axiomatic AI, Inc., 675 Massachusetts Ave, Cambridge, MA 02139}
\author{Alexey~V.~Gorshkov}
\affiliation{\JQI}
\affiliation{\QUICS}
\author{Jun Suzuki}
\email{junsuzuki@uec.ac.jp}
\affiliation{\UEC}
\author{Syed M. Assad}
\email{cqtsma@gmail.com}
\affiliation{\ASTAR}

\begin{abstract}
The quantum Fisher information (QFI) sets a fundamental bound on the attainable precision when estimating multiple parameters simultaneously. Incompatibility among the individually optimal measurements can, in some cases, imply that the precision limit set by the QFI is not attainable. In certain special cases, including pure states and full-rank states, the exact conditions for when the precision limit set by the QFI can be saturated are known. However, general conditions for the attainability of the QFI with measurements on individual copies of the quantum state have long been sought. Indeed, this was recently stated as one of the five problems in quantum information theory highlighted by [\mbox{P. Horodecki \textit{et al}}, \mbox{PRX Quantum} 3, 010101 (2022)]. In this work we prove that exact attainability with individual measurements is \np-hard to decide, even for a restricted set of real, constant-rank quantum models. The source problem for our proof is the \np-hard problem of deciding whether a given bipartite density matrix is separable or not. Our result shows that the longstanding difficulty in obtaining general conditions for QFI attainability reflects a fundamental computational obstruction, rather than merely a limitation of existing mathematical techniques: unless $\p=\np$, no efficiently computable necessary-and-sufficient criterion can exist in general.
\end{abstract}
\maketitle

\section{Introduction}
Measurements allow scientists to learn about the universe we live in, and as such constitute an invaluable part of the scientific process. Quantum mechanics simultaneously offers new opportunities for enhanced measurements~\cite{aasi2013enhanced,marciniak2021optimal,huang2022quantum,conlon2023approaching}, and also new restrictions on measurements through the uncertainty principle~\cite{robertson1929uncertainty,heisenberg1985anschaulichen,conlon2026100}. Quantum multiparameter estimation is a mathematical framework for analysing the optimal way to extract information from a physical quantum system, where the limitations imposed by the uncertainty principle are particularly pronounced. In this framework a quantum density matrix is described in terms of $p$ unknown parameters $\rho=\rho(\theta_1,...,\theta_p)$. By performing measurements and assigning a predicted value to each observed sequence of outcomes one can construct an estimate of each unknown parameter $\hat{\theta}_i$, whose uncertainty is limited by the covariance matrix of a given estimation procedure, $V$. The central aim of quantum multiparameter estimation is to minimise the mean squared error (MSE), $\Tr[V]$.

The value of $V$ is lower bounded by a quantity known as the quantum Fisher information (QFI) matrix $\mathcal{J}_\text{Q}$, via $V\geq\mathcal{J}_\text{Q}^{-1}$~\cite{helstrom1967minimum,helstrom1968minimum}. For minimisation, it is common to turn to the scalar quantity $\Tr[V]$ that satisfies $\Tr[V]\geq\Tr[\mathcal{J}_\text{Q}^{-1}]=\mathcal{C}_\text{Q}$, known as the quantum Cram{\'{e}}r-Rao bound (QCRB)\footnote{This is also referred to as the Helstrom bound or the symmetric logarithmic derivative Cram{\'{e}}r-Rao bound in the literature.}. Remarkably, in the single-parameter case, this inequality is tight and the optimal measurement and estimator satisfying $V_{11}=\mathcal{C}_\text{Q}=1/\mathcal{J}_\text{Q}$ are known~\cite{braunstein1994statistical} enabling optimal experiment design. However, when estimating multiple parameters the situation is more nuanced. For example, the optimal measurement for estimating one parameter $\theta_1$ may not be the same as the optimal measurement for estimating any other parameter $\theta_i$. When these optimal measurements do not commute, it may not be possible to achieve equality above, i.e. there may not exist a measurement for which $V=\mathcal{J}_\text{Q}^{-1}$. To aid attempts to design practical experiments with optimal precision it is desirable to understand when equality is possible in the above inequality. We shall refer to the existence of a measurement such that $V=\mathcal{J}_\text{Q}^{-1}$ as the condition for \textit{exact attainability} of the QCRB (Note that $\Tr[V-\mathcal{J}_\text{Q}^{-1}]=0$  implies $V=\mathcal{J}_\text{Q}^{-1}$ as $V-\mathcal{J}_\text{Q}^{-1}\succeq0$). This leads us to a longstanding open question in quantum parameter estimation, recently stated as one of five open problems in quantum information~\cite{horodecki2022five}: \textit{Under what conditions is the multiparameter QCRB exactly attainable with measurements on individual copies of the quantum state $\rho$?}

There has been great interest in this problem since Helstrom introduced the QCRB in 1967~\cite{helstrom1967minimum,helstrom1968minimum}. In 1973, Yuen and Lax introduced an alternative Cram{\'{e}}r-Rao bound that is sometimes tighter than QCRB~\cite{yuen1973}. Shortly after, Holevo introduced a bound that unified the above bounds, the Holevo Cram{\'{e}}r-Rao bound, $\mathcal{C}_\text{H}$~\cite{holevo1973statistical,holevo2011probabilistic}. Recently it has been shown that this is a tight bound which can be asymptotically approached through entangling measurements on infinitely many copies of the quantum state~\cite{kahn2009local, yamagata2013quantum, yang2019attaining}. Additionally, the conditions for $\mathcal{C}_\text{H}=\mathcal{C}_\text{Q}$ are known~\cite{ragy2016compatibility}. Taken together Refs.~\cite{kahn2009local, yamagata2013quantum, yang2019attaining} and \cite{ragy2016compatibility} provide the conditions under which the QCRB is attainable when one allows entangling measurements on infinitely many copies of the quantum state. However, owing to the practical difficulty in implementing entangling measurements on even a small number of copies of the quantum state~\cite{conlon2023discriminating,hou2018deterministic,yung2025saturating,yung2026beating}, it is still desirable to answer the above open problem.

Significant progress has been made towards attainability with measurements on individual copies. In 2002 Matsumoto answered the above question for pure states~\cite{matsumoto2002new} (see also Ref.~\cite{pezze2017optimal}). For mixed states the Nagaoka--Hayashi Cram{\'{e}}r-Rao bound, $\mathcal{C}_\text{NH}$, provides an efficiently computable lower bound on the precision attainable with measurements on individual copies~\cite{nagaoka2005generalization,nagaoka2005new,hayashi1999,conlon2021efficient}. However, this bound is not tight in general~\cite{hayashi2023tight,conlon2025role} and so cannot be used to answer the above question. Recently, Suzuki, Yang and Hayashi appear to have answered the above question in Appendix B.1 of Ref.~\cite{Suzuki2020}. They prove that the QCRB is attainable with measurements on individual copies of the quantum state if and only if the kernel space of a certain set of operators (to be introduced shortly) can be chosen such that the operators commute. However, this does not provide any constructive conditions, the difficulty has merely been transferred from one problem (attainability of the QCRB) to another (deciding whether the operators can be chosen to commute). Similarly, Hayashi and Ouyang have recently introduced a conic programming approach that proves that the QCRB is attainable with individual measurements if and only if an operator associated with the QCRB lies in the separable cone~\cite{hayashi2023tight}. Again, this does not provide a constructive condition for attainability of the QCRB.

Motivated by these observations, and the fact that this longstanding search has been unfruitful, we refine our question to ask: \textit{Does there exist an efficient algorithm for deciding whether the multiparameter QFI Cram{\'{e}}r-Rao bound is attainable with measurements on individual copies of the quantum state?} By an efficient algorithm we mean one whose runtime is polynomial in the total bit length of the matrices specifying the multiparameter estimation model. In this work, we answer this question in the negative, proving that deciding whether the QCRB is attainable is \np-hard. Therefore, provided $\p\neq\np$, we have proven that no such efficiently computable condition exists. While this does not directly answer the open question in Ref.~\cite{horodecki2022five}, this algorithmic complexity provides significant insight---demonstrating that the obstructions researchers have faced over the past half century to decide when the QCRB is attainable are fundamental rather than a deficiency of existing methods.

\section{Preliminary material}
In this section we first introduce quantum multiparameter estimation, before describing the source problem for our \np-hardness proof.

\subsection{Quantum multiparameter estimation}
We consider a $D$-dimensional Hilbert space $\mathcal{H}$. A quantum statistical model is defined by a density matrix $\rho(\theta_1,...,\theta_p)$ and the corresponding derivatives $\rho_i=\partial_i\rho$, where $\partial_i$ denotes derivative with respect to parameter $\theta_i$. 

To construct estimates of the unknown parameters, $\hat{\theta}_i$, one can implement a measurement described by a positive operator valued measure (POVM), which is a set of positive operators $\{\Pi_k\}$ that sum to the identity
\begin{equation}
\label{Eq:BGsumpi}
\sum_{k}\Pi_{k}=I_{D}\;.
\end{equation}
The $k$-th measurement outcome occurs with probability
\mbox{$p_k=\Tr[\rho\Pi_{k}]$} such that the MSE matrix is given by 
\begin{equation}
\label{eqMSEmatrix}
[V(\theta)]_{ij}=\sum_{k}(\hat{\theta}_{i}(k)-\theta_{i})(\hat{\theta}_{j}(k)-\theta_{j})p_k\;.
\end{equation}
In this work we focus on local estimation, where we only require that the estimators are unbiased to first order around a known value $\theta_0=(\theta_{0,1},\theta_{0,2},...,\theta_{0,p})$. The local unbiased conditions are that the expected value $\langle\hat{\theta}_i\rangle$  satisfies $\langle\hat{\theta}_i\rangle=\theta_{0,i}$ to first order in the Taylor series expansion~\cite{Fujiwara2006Adaptive}:
\begin{equation}
    \begin{split}
&\langle\hat{\theta}_i\rangle=\theta_{0,i}\;,\\
&\frac{\partial }{\partial\theta_j}\langle\hat{\theta}_i\rangle\bigg|_{\theta=\theta_0}=\delta_{ij}\;.
    \end{split}
\end{equation}
Without loss of generality we take $\theta_0=(0,0,...,0)$. Going forward, we use $\rho_0$ to denote $\rho$ evaluated at $\theta_0$ and $\rho_i$ to denote the corresponding derivative evaluated at $\theta_0$. In this setting we can define the most informative bound as
\begin{equation}
\label{eq:CMI}
\mathcal{C}_\text{MI}(\rho,\rho_i,W) = \inf_{\{\Pi_k\},\hat{\theta}}\Tr[WV(\theta)]\;,
\end{equation}
where the infimum is taken over all $\{\Pi_k\},\hat{\theta}$ satisfying the locally unbiased conditions and $W$ is a positive definite weight matrix. For the remainder of this work we set $W=I_p$ with no loss of generality\footnote{To see why the restriction to $W=I_p$ is sufficient, note that we are only concerned with equality conditions here and $\Tr[V-\mathcal{J}_\text{Q}^{-1}]=0$ combined with $V-\mathcal{J}_\text{Q}^{-1}\succeq0$ implies $V=\mathcal{J}_\text{Q}^{-1}$, and therefore, $\Tr[WV]=\Tr[W\mathcal{J}_\text{Q}^{-1}]$. }. Note that the infimum in Eq.~\eqref{eq:CMI} is not always attainable~\cite{yamagata2026sufficient}. When the bound is attainable, it corresponds to the MSE as the POVM and estimator are unbiased. 

No general efficient method
is known for finding the optimal measurement in the minimization
problem Eq.~\eqref{eq:CMI}. As such we turn to Cram{\'{e}}r-Rao bounds, such as the QCRB defined in the introduction, which provide lower bounds on the covariance of any locally unbiased estimator. To define the QCRB, we need to introduce the symmetric logarithmic derivative (SLD) operators as Hermitian operators that are implicitly defined through the equation $\rho_i=\frac12\left(L_i\rho+\rho L_i\right)\;$. The corresponding QFI matrix is then
\begin{equation}
\mathcal{J}_{\text{Q},ij}=\frac12\Tr[\rho(L_iL_j+L_jL_i)]\;,
\end{equation}
and the corresponding QCRB is given by
\begin{equation}
  \Tr[V]\geq\mathcal{C}_\text{Q}=\Tr[\mathcal{J}_\text{Q}^{-1}]\;.
\end{equation}
The central open question asked by many researchers for the past 60 years has been: under what conditions does there exist a locally unbiased POVM and estimator such that $V=\mathcal{J}_\text{Q}^{-1}$? Our results prove that the lack of solution is due to a fundamental computational bottleneck. From a computational perspective, we can rewrite the question of whether $V=\mathcal{J}_\text{Q}^{-1}$ as the following problem:
\begin{definition}[Exact attainability decision problem]\label{def:target}
The input consists of rational matrices $(\rho_0,\rho_1,\ldots,\rho_p)$ satisfying $\rho_0=\rho_0^\dagger\succeq0$, $\Tr[\rho_0]=1$, $\rho_i=\rho_i^\dagger$, and $\Tr[\rho_i]=0$, and for which the SLD equations are solvable and $\mathcal{J}_\text{Q}\succ0$.  The question is whether there exists a finite-outcome POVM and a locally unbiased estimator satisfying $V=\mathcal{J}_\text{Q}^{-1}$.
\end{definition}

As mentioned in the introduction, non-constructive conditions for when $V=\mathcal{J}_\text{Q}^{-1}$ have been presented in Ref.~\cite{Suzuki2020}. Observe that the equation $\rho_i=\frac12\left(L_i\rho+\rho L_i\right)\;$ uniquely determines the blocks of $L_i$ in the support and support-kernel subspaces of $\rho$. The elements of $L_i$ in the kernel space can be freely chosen provided $L_i$ is Hermitian. It is also possible to consider a finite extension of the model via an ancilla space $\mathcal{A}$, such that $\tilde{\mathcal{H}}=\mathcal{H}\oplus\mathcal{A}$, $\tilde{\rho}_0=\rho_0\oplus0_{\mathcal{A}}$, and $\tilde{\rho_i}=\rho_i\oplus0_{\mathcal{A}}$ (see Appendix~\ref{apen:extendedSLDs} and Ref.~\cite{conlon2025role} for more details). Such an extension allows more freedom in how the kernel of $L_i$ is chosen. With this in mind, we present the following theorem from Appendix B.1 of Ref.~\cite{Suzuki2020}:
\begin{theorem}[Suzuki, Yang, \& Hayashi.~\cite{Suzuki2020}]\label{theorem:finitecriterion}
   Assume $\mathcal{J}_{\text{Q}}\succ0$.  A finite-outcome locally unbiased estimator attains covariance $\mathcal{J}_{\text{Q}}^{-1}$ if and only if a valid extension of the SLD operators can be chosen such that $[L_i,L_j]=0\;,\;\forall i,j$.
\end{theorem}

\begin{figure}[t]
    \centering
    \includegraphics[width = 0.45\textwidth]{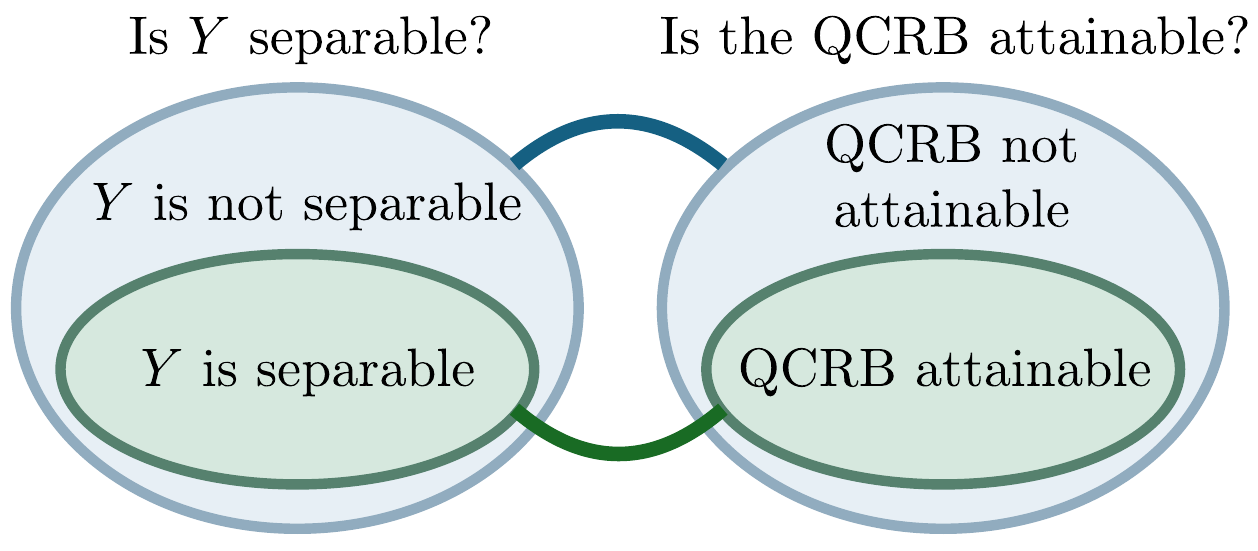}
    \caption{\textbf{Schematic depiction of main results.} The left and right outer blue ovals denote the set of all matrices $Y$ of a certain dimension and the corresponding set of multiparameter estimation problems defined in Theorem~\ref{th:quantum-encoding} respectively. In this work, we map the problem of deciding whether a given matrix $Y$ is separable or not to the problem of deciding whether the QCRB is attainable (Theorem~\ref{th:quantum-encoding}). As it is known that deciding whether a bipartite density matrix is separable or not is \np-hard (Theorem~\ref{thm:gurvits}), this implies that deciding whether the QCRB is attainable or not is also \np-hard (Theorem~\ref{th:nphard}).}
    \label{figschematic}
\end{figure}

\subsection{Hardness source problem}
\label{subsec:source}

We now introduce the source problem for our \np-hardness proof. Gurvits proved that the weak membership problem for the convex set of separable normalized bipartite density
matrices is \np-hard~\cite{gurvits2004classical}. We consider a slight modification of Gurvits's original work. Consider the set
\begin{align}
  \Sep^{\R}_{p,r}
  &=\operatorname{cone}\{zz^T\otimes ww^T:z\in\R^p,\ w\in\R^r\},
  \label{eq:sepcone}
\end{align}
where $\text{cone}\{S\}$ denotes all finite non-negative linear combinations of $S$. Note that Eq.~\eqref{eq:sepcone} is the real-separable cone. Ref.~\cite{gurvits2004classical} specifically considered weak membership of the complex separable cone, for our purposes we specialise to the real cone.

As we explain in Appendix~\ref{apen:Gurvits}, the following theorem follows almost directly from Ref.~\cite{gurvits2004classical}:
\begin{theorem}[Gurvits~\cite{gurvits2004classical}]\label{thm:gurvits}
Given a rational $Y \geq 0$, $Y\in\Sym_p(\R) \otimes \Sym_r(\R)$, with $\Tr[Y]=1$, and $p = r(r-1)/2+2$, the exact real-separability decision problem (deciding if $Y \in \Sep^{\R}_{p,r}$) is \np-hard under polynomial-time Turing reductions.
\end{theorem}
Above $\Sym_{x}(\R)$ denotes the real vector space of $x\times x$ real symmetric matrices.

\section{Main Results}
We are now in a position to give a brief outline of our main result. Our aim is to show that the decision problem in Definition~\ref{def:target} is \np-hard. We do this by showing that if one could construct an algorithm to decide this problem efficiently, then we could also solve a known \np-hard problem efficiently, as depicted in Fig.~\ref{figschematic}. Given a real positive semidefinite matrix $H\in\Sym_p(\R)\otimes\Sym_r(\R)$, with $\Tr[H]=1$, Theorem~\ref{thm:gurvits} tells us that deciding if $H\in\Sep^{\R}_{p,r}$ is \np-hard. Now assume we can efficiently construct a multiparameter estimation model given $H$, i.e.~a family of $\rho$, $\rho_i$, such that deciding if $H\in\Sep^{\R}_{p,r}$ is equivalent to deciding whether a valid extension of the corresponding SLD operators $L_i$ exists such that the SLD operators commute. The existence of such a model would prove that deciding whether the QCRB is attainable or not is \np-hard. We now present an explicit multiparameter estimation model that does exactly this:
\begin{theorem}[Quantum encoding]\label{th:quantum-encoding}
Given a rational positive semidefinite $H\in\Sym_p(\R)\otimes\Sym_r(\R)$,  one can construct, in polynomial time, rational local data $(\rho_0,\rho_1,\ldots,\rho_p)$ of a  real constant-rank model with $\mathcal{J}_\text{Q}\succ0$ such that $H\in\Sep^{\R}_{p,r}$ if and only if some finite extension admits commuting SLD representatives. 
\end{theorem}
This will be proved in Appendix~\ref{apen:quantummodel}. Here, polynomial time means polynomial in $p$, $r$, and the maximum bit length of the rational entries of $H$. For completeness, we state the hard model explicitly here. Choose a rational factorization $H=BB^T$ of $H$, for a rational rectangular matrix $B$ whose number of columns is polynomial in the length of the input $H$ (Lemma~\ref{lem:bbt}). Partition $B$ into row-blocks as $B=\left(\begin{smallmatrix}B_1\\\vdots\\B_p\end{smallmatrix}\right)$. Let $\mathcal T\cong\C^{p+1}$ have basis $|0\rangle,\ldots,|p\rangle$ and define
\begin{equation}
  T_i=|i\rangle\langle i|-|0\rangle\langle0|,
  \qquad 1\le i\le p.
\end{equation}
Let $\mathcal S\cong\C^r$, which will encode the \np-hard problem (i.e. it carries the information about the input $H$), and let $\mathcal{N}\cong\C^q$ contain the freedom of the SLD operators (i.e. the kernel of $\rho$). On $\mathcal H=(\mathcal T\oplus\mathcal S)\oplus \mathcal{N}$, put $r'=p+1+r$ and define
\begin{equation}
  L_i^{(0)}=
  \begin{pmatrix}
    T_i\oplus0_{\mathcal S}&\binom{0}{B_i}\\[2mm]
    (\,0\ \ B_i^T\,)&0_\mathcal{N}
  \end{pmatrix}\;.
\end{equation}
Defining $X(\theta)=\sum_i\theta_iL_i^{(0)}-\psi(\theta)I$, where $\psi(\theta)$ is a normalisation factor, we arrive at the smooth quantum model that is used in Appendix~\ref{apen:quantummodel} to prove Theorem~\ref{th:quantum-encoding},
\begin{equation}
  \rho_\theta=
\mathrm{e}^{\frac12X(\theta)}\rho_0\mathrm{e}^{\frac12X(\theta)}\;,
\end{equation}
with $\rho_0=(I_{r'}\oplus 0_{q})/r'$.

We can now state our main result:
\begin{theorem}[Deciding the attainability of the QCRB is \np-hard]\label{th:nphard}
Solving the decision problem in Definition~\ref{def:target}, i.e. deciding whether $V=\mathcal{J}_\text{Q}^{-1}$ or not, is \np-hard under polynomial-time Turing reductions.
\end{theorem}
\begin{proof}
We reduce from the real separability problem of Theorem~\ref{thm:gurvits}.
Suppose we have an oracle $O_\text{att}$ that correctly decides whether $V=\mathcal{J}_\text{Q}^{-1}$ given any rational quantum multiparameter estimation problem satisfying Definition~\ref{def:target}. We show that this oracle can be used to decide if $Y\in\Sep^{\R}_{p,r}$ in Theorem~\ref{thm:gurvits}. The input is a rational query: a rational matrix $Y\in\Sym_{p}(\R)\otimes\Sym_r(\R)$ with $\Tr[Y]=1$ and $Y\succeq0$.

Set $H=Y$, and construct the model in Theorem~\ref{th:quantum-encoding}, which can be done in polynomial time. Theorem~\ref{theorem:finitecriterion} then states that the compatibility of the corresponding SLD operators is equivalent to  $V=\mathcal{J}_\text{Q}^{-1}$, which can be decided by the oracle $O_\text{att}$. If the oracle answers yes, then Theorem~\ref{theorem:finitecriterion} guarantees a commuting SLD extension exists, and then Theorem~\ref{th:quantum-encoding} implies $Y\in\Sep^{\R}_{p,r}$ and so we can give a positive answer to the separable query.

On the other hand, if the oracle answers no, then again using Theorems~\ref{theorem:finitecriterion} and \ref{th:quantum-encoding}, we see we must have $Y\not\in\Sep_{p,r}^\R$.
Thus we can output the negative answer.

Therefore, we have a polynomial-time reduction from the known \np-hard separability problem in Theorem~\ref{thm:gurvits} to the QCRB attainability problem in Definition~\ref{def:target}, and so we conclude the latter is also \np-hard.
\end{proof}
This result leads us to the following corollary
\begin{corollary}
Assuming $\p\neq\np$, there does not exist a polynomial time algorithm capable of deciding the attainability of the QCRB.
\end{corollary}
This result therefore immediately implies that open problem number 3 in Ref.~\cite{horodecki2022five} cannot be solved in an efficient manner unless $\p=\np$. It is also noteworthy that this result provides an interesting connection between entanglement witnesses and quantum metrology, beyond what is already known~\cite{hyllus2012fisher}.

The result above establishes that deciding whether the QCRB is exactly attainable is \np-hard, but by itself it gives no quantitative control over the NO instances. In particular, the reduction could in principle produce non-attainable models for which locally unbiased estimators can approach the QCRB arbitrarily closely. To address this, in Appendix~\ref{apen:approxattain}, we show that the hardness is robust: it remains \np-hard to distinguish models for which the QCRB is exactly attainable from models for which every locally unbiased estimator satisfies
\begin{equation}
\label{eq:polypr}
\Tr\left[V-\mathcal{J}_{\mathrm Q}^{-1}\right]
\geq \frac{1}{\operatorname{poly}(p,r)}.
\end{equation}
Moreover, for the models produced by our reduction for the promise problem of distinguishing exact attainability from Eq.~\eqref{eq:polypr}, every exactly attainable instance admits a certificate of polynomial size whose validity can be checked in polynomial time.

We note that the restriction to real rational models does not limit the validity of our results. Any algorithm capable of solving the decision problem of whether $V=\mathcal{J}_\text{Q}^{-1}$ for arbitrary inputs would also solve the hard rational model introduced here. The decision problem in Definition~\ref{def:target} naturally provides a relationship between QCRB attainability and the feasibility of systems of quadratic equations over the reals, see Appendix~\ref{apen:existence} and Ref.~\cite{schaefer2026existential}. Additionally, the restriction to POVMs with finite outcomes does not affect the generality of our results, see Appendix~\ref{sec:finite-support}. Finally, note that in Appendix~\ref{apen:NHCRHCRB} we prove that deciding whether there exists a POVM and estimator such that $\Tr[V]=\mathcal{C}_\text{NH}$ or $\Tr[V]=\mathcal{C}_\text{H}$ is also \np-hard.

\section{Discussion}
It has been a long-standing open question in quantum information to present conditions under which $V=\mathcal{J}_\text{Q}^{-1}$~\cite{horodecki2022five}. Indeed, this is still a very active research direction~\cite{imai2026hierarchy}. We have proven that no such efficiently computable condition exists unless $\p=\np$. As such, a complete answer to open problem 3 in Ref.~\cite{horodecki2022five} is unlikely to exist beyond that already presented by Suzuki, Yang and Hayashi~\cite{Suzuki2020}. This result is consistent with the decades of progress since Helstrom introduced his bound~\cite{helstrom1967minimum,helstrom1968minimum}: many efficiently computable lower bounds to $\mathcal{C}_\text{MI}$ have been introduced, however all known bounds merely approximate $\mathcal{C}_\text{MI}$~\cite{yuen1973,holevo1973statistical,nagaoka2005generalization,nagaoka2005new,hayashi1999,conlon2021efficient}.

One may wonder whether our result follows directly from existing literature. It does not. Suzuki, Yang, and Hayashi proved that the QCRB is attainable if and only if the SLD operators admit mutually commuting extensions on a sufficiently large Hilbert space~\cite{Suzuki2020}. Their condition is exact, but does not determine the computational complexity of deciding whether the unconstrained kernel blocks can be chosen to make the SLD operators commute. Similarly, Hayashi and Ouyang formulated the exact single-copy estimation problem as a conic program over the separable cone~\cite{hayashi2023tight}. Although general membership and optimization problems over the separable cone are \np-hard, this fact alone does not imply hardness of QCRB attainability, because the conic instances generated by valid quantum statistical models form a highly structured subset of arbitrary separability instances. The missing ingredient is a reduction in the opposite direction, showing that arbitrary hard separability instances can be encoded into valid quantum models with polynomial overhead. Our construction provides precisely this encoding and thereby establishes the \np-hardness of the physical attainability problem itself.

We note that for a fixed number of unknown parameters Hayashi and Ouyang have recently presented an algorithm for computing $\mathcal{C}_\text{MI}$ to within error $\epsilon$~\cite{hayashi2023tight}. The runtime of this algorithm contains an $\epsilon$-dependence of $\epsilon^{-p}$.\footnote{Therefore, given a required number of bits of precision $b=\log_2(1/\epsilon)$, the runtime carries a dependence of $2^{bp}$.} Note however, that such an algorithm cannot be used to determine whether $V=\mathcal{J}_\text{Q}^{-1}$. For example if one had determined via this algorithm that $\Tr[V-\mathcal{J}_\text{Q}^{-1}]<\epsilon$, we cannot conclude that $V=\mathcal{J}_\text{Q}^{-1}$. Our results demonstrate that this is a fundamental limitation and not a limitation of their chosen algorithm. Their algorithm can be used to solve the approximate attainability problem discussed in Appendix~\ref{apen:approxattain}. It is worth noting that efficient iterative methods for
searching over measurement strategies are available~\cite{kimizu2024adaptive,zhang2024qestoptpovm}, although they do
not generally certify convergence to a global optimum.

Although this work closes a major open problem in quantum multiparameter estimation, there remain many interesting avenues along which we may hope for parallel advances. In the light of our results it becomes particularly important to search for multiparameter estimation models where it is efficiently computable to determine whether $V=\mathcal{J}_\text{Q}^{-1}$ or not. Another important avenue is to extend our results to global estimation, i.e. does there exist a multiparameter estimation model where it is hard to decide the attainability of the QCRB for all possible values of the unknown parameters?

More broadly, the quantum Fisher information has become a unifying quantity across quantum information theory. Notably, it has relevance for connecting statistical distinguishability and information geometry~\cite{braunstein1994statistical}, multipartite entanglement detection~\cite{hyllus2012fisher,hauke2016measuring}, quantum criticality~\cite{zanardi2008quantum}, quantum speed limits~\cite{taddei2013quantum}, and variational quantum algorithms~\cite{meyer2021fisher}. Our result exposes a fundamental computational limitation: although this information-geometric benchmark is efficiently computable, deciding whether its ultimate precision can be physically attained is \np-hard.

\begin{acknowledgments}
ChatGPT 5.6 Sol was used in the proof of the main theorem of this work. Axiomatic AI’s Lemma was used to develop examples and validate proofs. All proofs were checked by the authors.

L.O.C. and J.M.T. were supported in part by the National Science Foundation under Award No. 2533041 (NQVL:Design:ORAQL). L.O.C., L.S., and A.V.G.~were supported in part by ONR MURI, AFOSR MURI, NSF QLCI (award No.~OMA-2120757),  NSF STAQ program, DoE ASCR Quantum Testbed Pathfinder program (award No.~DE-SC0024220),  ARL (W911NF-24-2-0107), DARPA SAVaNT ADVENT, and NQVL:QSTD:Design:FTL. L.O.C., L.S., and A.V.G.~also acknowledge support from the U.S.~Department of Energy, Office of Science, National Quantum Information Science Research Centers, Quantum Systems Accelerator (award No.~DE-SCL0000121) and from the U.S.~Department of Energy, Office of Science, Accelerated Research in Quantum Computing, Fundamental Algorithmic Research toward Quantum Utility (FAR-Qu). V.V. is supported by the National Research Foundation, Singapore, through the National Quantum Office, hosted in A*STAR, under its National Quantum Scholarship Scheme (NQSS) Funding Initiative (PhD). J.S. was partly supported by JSPS KAKENHI Grant Numbers JP24K14816 and ERATO ``Super Quantum Entanglement" (Grant No. JPMJER2402) from JST. K.B.~was supported by a Hartree Fellowship from the Joint Center for Quantum Information and Computer Science (QuICS) at the University of Maryland, College Park.

\end{acknowledgments}

\bibliography{references}
\bibliographystyle{naturemag}

\clearpage
\appendix

\onecolumngrid
\makeatletter
\onecolumn@grid@setup
\let\set@footnotewidth\set@footnotewidth@one
\let\compose@footnotes\compose@footnotes@one
\makeatother

\numberwithin{theorem}{section}

\begin{center}
\textbf{\large Supplemental Material}
\end{center}

Here we provide the full proof details of Theorem~\ref{th:nphard} as well as other information that may be useful. We begin by discussing the extended SLD operators in Appendix~\ref{apen:extendedSLDs}. In Appendix~\ref{apen:Gurvits} we present the proof of Theorem~\ref{thm:gurvits}, followed by the proof of Theorem~\ref{th:quantum-encoding} in Appendix~\ref{apen:quantummodel}. In Appendices~\ref{apen:rationaldecomp}, \ref{apen:separability}, and \ref{apen:realization} we present three results that are used in the proof of Theorem~\ref{th:quantum-encoding} in Appendix~\ref{apen:quantummodel}.

\addtocontents{toc}{\protect\supplementtocstart}
\begingroup
\makeatletter
\let\savedl@section\l@section
\let\l@section\@gobbletwo
\let\l@subsection\@gobbletwo 
\let\l@subsubsection\@gobbletwo 
\renewcommand{\supplementtocstart}{%
  \let\l@section\savedl@section}
\tableofcontents
\makeatother
\endgroup

\section{Extended SLD operators}
\label{apen:extendedSLDs}
In this appendix we describe how the kernel of $\rho$ can be extended without affecting the QFI or QCRB. This extension provides greater freedom when choosing the kernel elements of the SLD operators and is relevant to Theorem~\ref{theorem:finitecriterion}. As in the main text, we consider a $D$ dimensional Hilbert space $\mathcal{H}$. Given a density matrix $\rho_0$, we partition $\mathcal{H}$ as $\mathcal{H}=\mathcal{H}_\text{S}\oplus\mathcal{H}_\text{K}$ where $\mathcal{H}_\text{S}$ is the support of $\rho_0$ and $\mathcal{H}_\text{K}$ is the kernel of $\rho_0$~\cite{conlon2025role}. Working in this basis we can write $\rho_0$ as
\begin{equation}
    \rho_0=\begin{pmatrix}
        \rho_{0,\text{S}}&0\\
        0&0
    \end{pmatrix}\;,
\end{equation}
where $\rho_{0,\text{S}}\succ0$, and its derivatives as
\begin{equation}
  \rho_i=\begin{pmatrix}
        D_i&E_i\\E_i^\dagger&0
    \end{pmatrix}\;.
\end{equation}
A general SLD operator can be written as
\begin{equation}
    L_i=\begin{pmatrix}
        L_{\text{S},i}&L_{\text{SK},i}\\
        L_{\text{SK},i}^\dagger &L_{\text{K},i}
    \end{pmatrix}\;.
\end{equation}
Therefore, the equation $\rho_i=\frac12\left(L_i\rho+\rho L_i\right)\;$ becomes
\begin{equation}
    2\begin{pmatrix}
        D_i&E_i\\E_i^\dagger&0
    \end{pmatrix}=\begin{pmatrix}
        \rho_{0,\text{S}}L_{\text{S},i}&\rho_{0,\text{S}}L_{\text{SK},i}\\
     0 &0
    \end{pmatrix}+\begin{pmatrix}
        L_{\text{S},i}\rho_{0,\text{S}}&0\\
        L_{\text{SK},i}^\dagger\rho_{0,\text{S}} &0
    \end{pmatrix}\;.
\end{equation}
As such, we see that $L_{\text{S},i}$ and $L_{\text{SK},i}$ are completely defined by the SLD equation, but $L_{\text{K},i}$ is unspecified.

Observe also that changing $L_{\text{K},i}$ does not affect the QFI as
\begin{equation}
    \mathcal{J}_{\text{Q},ij}=\frac12\Tr[\rho(L_iL_j+L_jL_i)]=\frac12\Tr\left[\rho_{0,\text{S}}\left(\{L_{\text{S},i},L_{\text{S},j}\}+L_{\text{SK},i}L^\dagger_{\text{SK},j}+L_{\text{SK},j}L^\dagger_{\text{SK},i}\right)\right]\;.
\end{equation}

Additionally note that we can consider a finite zero-weight extension that does not change any of the above arguments but increases the dimension of the kernel space:
\begin{equation}
    \tilde{\mathcal{H}}=\mathcal{H}\oplus \mathcal{A}\;,\tilde{\rho}_0=\rho_0\oplus 0_{\mathcal{A}}\;,\tilde{\rho}_i=\rho_i\oplus 0_{\mathcal{A}}.
\end{equation}
This extended kernel space provides additional degrees of freedom when choosing $L_{\text{K},i}$. This in turn can help when choosing $L_{\text{K},i}$ such that the SLD operators commute, as in Theorem~\ref{theorem:finitecriterion}. Indeed, it is known that for certain problems this extended space is necessary for choosing commuting SLD operators~\cite{conlon2025role}.

\section{Extension of Gurvits (Theorem~\ref{thm:gurvits})}
\label{apen:Gurvits}
In this appendix we describe in more detail the source problem for our \np-hardness proof introduced in section~\ref{subsec:source} (Theorem~\ref{thm:gurvits}). We consider the set
\begin{align}
  \mathcal K^{\R}_{p,r}
  &=\operatorname{conv}\{zz^T\otimes ww^T:\label{eq:sepbase} z\in\R^p,w\in\R^r,\norm{z}_2=\norm{w}_2=1\}\;,
\end{align}
where $\text{conv}\{S\}$ denotes the convex hull of $S$, i.e.~all finite convex combinations of elements of $S$. Eq.~\eqref{eq:sepbase} is the normalised base of the real separable cone in Eq.~\eqref{eq:sepcone}. We now consider weak membership of this cone.

Weak membership of a set $K$ is a promise problem\footnote{A promise problem can be viewed as a generalization of a decision problem which only requires correct output on a certain set of inputs (the ``promised'' inputs). In this case, one does not need to determine membership of points near the boundary of $K$.} which asks: given a rational matrix $Y$, and a rational number $\delta>0$, either:
\begin{enumerate}
    \item Assert that $Y\in S(K,\delta)$, or
    \item Assert that $Y\notin S(K,-\delta)$\;,
\end{enumerate}
where 
\begin{align}
\label{eq:defweaksep}
  S(K,\delta)&=\{y\in\text{aff}(K):\dist(y,K)\le\delta\},\\
  S(K,-\delta)&=\{x\in K:x+B(0,\delta)\subseteq K\},
\end{align}
where $\dist(y,K)=\inf_{X\in K}||y-X||$ where the norm is Frobenius norm, $\text{aff}(K)$ is the affine hull of $K$, and $B(0,\delta)$ is the closed radius-$\delta$ ball. $B(0,\delta)$ is defined as:
\begin{equation}
  B(0,\delta)=  \{z\in \operatorname{aff}(K)-\operatorname{aff}(K):
|z|_{\mathrm F}\le\delta\}\;,
\end{equation}
where $\operatorname{aff}(K)-\operatorname{aff}(K)$ is the linear subspace parallel to the affine hull of \(K\).

To prove Theorem~\ref{thm:gurvits}, we use that weak membership in $\mathcal K^{\R}_{p,r}$ is \np-hard \cite{gurvits2004classical}. For this, we first verify that $\mathcal K^{\R}_{p,r}$ is compact. Recall $\mathcal K^{\R}_{p,r}$ is the convex hull of $G_{p,r}:=\{zz^T\otimes ww^T:\norm{z}_2=\norm{w}_2=1\}$, i.e. it is the set of all convex combinations $\sum_{i=1}^m \lambda_ig_i$, for points $g_i\in G_{p,r}$ and real $\lambda_i\ge0$ with $\sum_i\lambda_i=1$. Since the convex hull of a compact set in $\R^d$ is compact \cite[Corollary~(2.4)]{barvinok2002course}, we just need to check that $G_{p,r}$ is compact. This follows since $G_{p,r}$ is the image of the compact set $\{(z,w)\in\R^p\times\R^r:\|z\|_2=\|w\|_2=1\}$ under the map $F(z,w):=zz^T\otimes ww^T$, which is continuous as every entry of $F(z,w)$ is a polynomial in the entries of $z$ and $w$.

To prove Theorem~\ref{thm:gurvits}, note that
\begin{equation}
  \Sep^{\R}_{p,r}
  =\{tX:t\ge0,\ X\in\mathcal K^{\R}_{p,r}\}.
  \label{eq:conebase}
\end{equation}
Hence for $\Tr [Y]=1$,
\begin{equation}
  Y\in\Sep^{\R}_{p,r}
  \quad\Longleftrightarrow\quad
  Y\in\mathcal K^{\R}_{p,r}.
  \label{eq:traceone}
\end{equation}
Theorem~\ref{thm:gurvits} then follows from Gurvits's Definition~6.2, Theorem~6.7, and Remark~6.8 \cite{gurvits2004classical}; see also Ref.~\cite{ioannou2007computational} for further discussion on the RSDF (robust semidefinite feasibility) dimension requirements.

An exact separability answer thus gives a valid weak-membership answer.  If a trace-one query $Y$ belongs to $\Sep^{\R}_{p,r}$, then \eqref{eq:traceone} gives $Y\in\mathcal K^{\R}_{p,r}\subseteq S(\mathcal K^{\R}_{p,r},\delta)$.  If it does not belong to the cone, then $Y\notin\mathcal K^{\R}_{p,r}$ and therefore $Y\notin S(\mathcal K^{\R}_{p,r},-\delta)$. A non-positive-semidefinite query can be rejected before the quantum construction because every real-separable tensor is positive semidefinite.

\section{Explicit quantum model (Theorem~\ref{th:quantum-encoding})}
\label{apen:quantummodel}
In this appendix we prove Theorem~\ref{th:quantum-encoding}. This result is based on three smaller results which we state here and prove in the subsequent appendices. 

We begin with a lemma concerning the decomposition $H=BB^T$. Because we work with a Turing machine in the reduction, we want to work with rational numbers only, and moreover rational numbers which can be expressed using a polynomial number of bits.
We need to factor $H=BB^T$ for use in Theorem~\ref{thm:sep}, but the usual Cholesky factorization can involve irrational numbers. Therefore, we allow polynomially-larger rectangular blocks $B$ in order to obtain a rational factorization $H=BB^T$ in polynomial time and with polynomial bit precision.
\begin{lemma}[rational factorization]\label{lem:bbt}
Let $H\ge0$ be an $n\times n$ matrix with rational entries. Then there is a rational matrix $B$ with polynomially many columns and whose entries have polynomially many bits, and which satisfies $H=BB^T$. Additionally, $B$ is computable in polynomial time.
\end{lemma}

We will prove Lemma~\ref{lem:bbt} in Appendix~\ref{apen:rationaldecomp}.
We next provide the condition for when $H$ is a real separable matrix.
\begin{theorem}[separability condition]\label{thm:sep}
Fix $p,r\in\N$, and let $H\in\Sym_p(\R)\otimes\Sym_r(\R)$ be a $pr\times pr$ positive semidefinite matrix. We can naturally partition $H$ into submatrices, $H=(H_{\alpha\beta})_{\alpha,\beta=1}^p$, for $r\times r$ submatrices $H_{\alpha\beta}$. Suppose we have a decomposition $H=BB^T$ for some $pr\times n_B$ rectangular matrices $B$ (cf. Lemma~\ref{lem:bbt}), with $n_B$ polynomial in $pr$ and in the number of bits required to represent entries in $H$. Partition $B$ into $r\times n_B$ row-blocks, $B=\left(\begin{smallmatrix}B_1\\\vdots\\B_p\end{smallmatrix}\right)$, so that $H_{\alpha\beta}=B_\alpha B_\beta^T$.

Then $H$ is real separable if and only if the $B_\alpha$ admit a finite commuting Hermitian completion; i.e., if and only if there exist Hermitian matrices $K_\alpha$ of some finite size so that the following matrices pairwise commute:
\begin{align}\label{eqn:Ca}
C_\alpha(K_\alpha):=\begin{pmatrix}
\mathbf 0_r& [B_\alpha\; \mathbf 0]\\
[B_\alpha\;\mathbf 0]^\dagger& K_\alpha
\end{pmatrix},\quad \alpha=1,\ldots,p,
\end{align}
where $[B_\alpha\;\mathbf 0]$ indicates padding with some finite number of zero columns. Moreover, when $H$ is real separable, then $C_\alpha$ can be taken real.
\end{theorem}
Theorem~\ref{thm:sep} is proven in Appendix~\ref{apen:separability}.

The last subingredient realizes the matrices in \eqref{eqn:Ca} as an explicit quantum multiparameter estimation model.
\begin{proposition}[Quantum-model realization]\label{lem:quantum-model}
Given a rational positive semidefinite $H\in\Sym_p(\R)\otimes\Sym_r(\R)$ and a rational factorization $H=BB^T$, one can construct a local quantum model specified by rational matrices $(\rho_0,\rho_1,\ldots,\rho_p)$ with $\mathcal{J}_\text{Q}\succ0$ such that finite commuting SLD representatives exist if and only if the commuting completion in Theorem~\ref{thm:sep} exists.  
\end{proposition}

Combining the above, we can now prove Theorem~\ref{th:quantum-encoding}. 

\begin{proof}[Proof of Theorem~\ref{th:quantum-encoding}]
Given $H$, we can compute a rational factorization $H=BB^T$ by Lemma~\ref{lem:bbt}. Applying Proposition~\ref{lem:quantum-model}, we obtain rational local data $(\rho_0,\rho_1,\ldots,\rho_p)$ with $\mathcal{J}_\text{Q}\succ0$ such that finite commuting SLD representatives exist if and only if the completion in Theorem~\ref{thm:sep} exists. Theorem~\ref{thm:sep} implies such a completion exists if and only if $H\in\Sep_{p,r}^\R$.
Polynomial time and encoding lengths are guaranteed by Lemma~\ref{lem:bbt} and Proposition~\ref{lem:quantum-model}.
\end{proof}

\section{Proof of rational decomposition (Lemma~\ref{lem:bbt})}
\label{apen:rationaldecomp}

In this appendix, we prove the rational decomposition $H=BB^T$ in Lemma~\ref{lem:bbt}.

\begin{proof}[Proof of Lemma~\ref{lem:bbt}]
The input is $H$, so all polynomial statements mean polynomial in $n$ and the number of bits required for the entries of $H$. 
Since $H$ is symmetric positive-semidefinite, we can use symmetric Gaussian elimination (i.e. applying matching row and column operations) to write \cite[\S4]{golub2013matrix}
\begin{align}
H=EDE^T,\quad D=\operatorname{diag}(d_1,\ldots,d_s,0,\ldots,0),
\end{align}
for positive rational entries $d_1,\ldots,d_s$ and elementary matrix operations matrix $E$.
We just need to take a (rational) square root of $D$ to complete the factorization.
For $d_j=a_j/b_j$ with $a_j,b_j\in\Z_+$ relatively prime, write $d_j=a_jb_j/b_j^2$. We can write $a_jb_j$ as a sum of $\ell_j:=O(\log(a_jb_j))$ squares by just expanding it in its binary expansion. We denote by $L$ the maximum of all $\ell_j$. Then writing $d_j=\sum_{\gamma=1}^{\ell_j}\frac{s_{j,\gamma}^2}{b_j^2}=\sum_{\gamma=1}^{\ell_j} q_{j,\gamma}^2$ and letting $v_j$ be the $j$th column of $E$, we get
\begin{align}
H&=\sum_{j=1}^s\sum_{\gamma=1}^{\ell_j} (v_jq_{j,\gamma})(q_{j,\gamma}v_j)^T.
\end{align}
Letting $B$ be the $n\times O(sL)$ matrix with columns $q_{j,\gamma}v_j$, indexed by $(j,\gamma)$, then immediately we get $H=BB^T$.

Note that $B$ has polynomially many columns and is computable in polynomial time via Gaussian elimination. Additionally, all entries have polynomial-length bit encodings by \cite[(1.4.8)/Edmonds' theorem]{grotschel1993geometric}.
(To see this directly, one can express the diagonal entries $d_j$ and the entries of $E$ using Schur complement, and bound the size of the determinants using Hadamard's inequality \cite[9.66]{axler2024linear}.)
\end{proof}

\section{Proof of separability condition (Theorem~\ref{thm:sep})} \label{apen:separability}
In this appendix, we prove the separability condition in Theorem~\ref{thm:sep}. We present two simple examples that illustrate this theorem.

\begin{proof}
$(\Leftarrow)$: We first prove the easier direction, that a commuting completion implies real separability. We suppose there are matrices $K_\alpha$ so that the $C_\alpha(K_\alpha)$ in \eqref{eqn:Ca} commute. Our goal is to show $H$ is $(p,r)$ real separable, i.e. can be written $H=\sum_j |x_j\rangle\langle x_j|\otimes|w_j\rangle\langle w_j|$ for $|x_j\rangle\in \R^p$ and $|w_j\rangle\in\R^r$. 
First, note that we can recover $H$ from the $C_\alpha=C_\alpha(K_\alpha)$ as follows. We see 
\begin{align}\label{eqn:CC}
C_\alpha C_\beta&=\begin{pmatrix}
B_\alpha B_\beta^T & [B_\alpha\;\mathbf 0]K_\beta\\
K_\alpha[B_\beta\;\mathbf 0]^T& [B_\alpha\;\mathbf 0]^T[B_\beta\;\mathbf 0]+K_\alpha K_\beta
\end{pmatrix}
\quad\Longrightarrow\quad H_{\alpha\beta}=B_\alpha B_\beta^T=PC_\alpha C_\beta P^T,
\end{align}
for $P$ the coordinate projection onto the first $r$ coordinates.
Since the $C_\alpha$ are Hermitian and commute, they have a common eigenbasis $\{|v_j\rangle\}_j$, and we can write $C_\alpha=\sum_jx_{\alpha,j}|v_j\rangle\langle v_j|$, for some real $x_{\alpha,j}$ and $\alpha=1,\ldots,p$. Inserting this into \eqref{eqn:CC} and letting $|u_j\rangle:=P|v_j\rangle$ gives
\begin{align}
H_{\alpha\beta}=\sum_j x_{\alpha,j}x_{\beta,j}|u_j\rangle\langle u_j|.
\end{align}
Letting $|x_j\rangle:=(x_{1,j},x_{2,j},\ldots,x_{p,j})$ so that $\langle\alpha|x_j\rangle=x_{\alpha,j}$, the above immediately gives $H=\sum_j|x_j\rangle\langle x_j|\otimes|u_j\rangle\langle u_j|$, which is of the desired separable form.
Since $|u_j\rangle$ may not be real, write $|u_j\rangle=|y_j\rangle+i|z_j\rangle$ for real $|y_j\rangle,|z_j\rangle$. Then the imaginary part must be 0 since $H$ is real, hence we get the real decomposition
\begin{align}
H&=\sum_j|x_j\rangle\langle x_j|\otimes (|y_j\rangle\langle y_j|+|z_j\rangle\langle z_j|).
\end{align}

$(\Rightarrow)$: Given a real separable $H$, we now construct a commuting extension $\{C_\alpha(K_\alpha)\}$. 
Note that in order for the $C_\alpha$ to commute, we will need the decomposition as in the other proof direction, $C_\alpha=\sum_jx_{\alpha,j}|v_j\rangle\langle v_j|$, with 
\begin{align*}
C_\alpha=\begin{pmatrix}\mathbf 0_r&[B_\alpha\;\mathbf 0]\\ [B_\alpha\;\mathbf 0]^\dagger&*\end{pmatrix},\quad\text{which gives the requirement}\quad PC_\alpha P^T=\sum_jx_{\alpha,j}|u_j\rangle\langle u_j|=\mathbf0_r,
\end{align*}
for $P$ the projection onto the first $r$ coordinates and $|u_j\rangle=P|v_j\rangle$. Additionally, we would need $\sum_j|u_j\rangle\langle u_j|=P\sum_j|v_j\rangle\langle v_j|P^T=I_{r}$.
Thus to construct such $C_\alpha$, we are looking for $|x_j\rangle$ and $|u_j\rangle$ satisfying the conditions:
\begin{align}\label{eqn:c-cond}
\sum_jx_{\alpha,j}|u_j\rangle\langle u_j|=\mathbf0_r,
\quad \sum_{j}|u_j\rangle\langle u_j|=I_r,\quad \sum_jx_{\alpha,j}x_{\beta,j}|u_j\rangle\langle u_j|=H_{\alpha\beta}.
\end{align}

Given some finite decomposition $H=\sum_i|z_i\rangle\langle z_i|\otimes|w_i\rangle\langle w_i|$, we will use the $z_i$'s and $w_i$'s to construct the desired $|x_j\rangle\in\R^p$ and $|u_j\rangle\in\R^r$. 
Since we can rescale $(z_i,w_i)\mapsto(tz_i,t^{-1}w_i)$ without changing $H$, we will take $t>0$ large enough so that $\sum_i |w_i\rangle\langle w_i|\preceq I_r$. The blocks of $H$ are $H_{\alpha\beta}=\sum_i z_{\alpha,i}z_{\beta,i}|w_i\rangle\langle w_i|$. To construct $|x_i\rangle=(x_{1,i},\ldots,x_{p,i})$ and $|u_i\rangle$ satisfying \eqref{eqn:c-cond}, define 
\begin{align}
|x_{i,\pm}\rangle:=\pm |z_i\rangle,\quad |u_{i,\pm}\rangle:=\frac{|w_i\rangle}{\sqrt{2}}.
\end{align}
Then letting $j$ run over $(i,\pm)$ gives the first and third equations of \eqref{eqn:c-cond}.
The second equation at the moment gives $\sum_j|u_j\rangle\langle u_j|=\sum_i|w_i\rangle\langle w_i|$.
We will add more vectors in order to change this into $I_r$. Recall we chose $t>0$ large enough so that $Q:=I_r-\sum_i|w_i\rangle\langle w_i|\succeq0$. Since $Q\succeq0$, factor it as $Q=\sum_k |q_k\rangle\langle q_k|$, and add the pairs $(|x_k\rangle,|u_k\rangle)=(\mathbf0,|q_k\rangle)$, for each $k$, to $(|x_j\rangle,|u_j\rangle)_j$. Since $|x_k\rangle=\mathbf0$, this doesn't affect the first or third equations of \eqref{eqn:c-cond}. By construction, the second equation $\sum_j|u_j\rangle\langle u_j|=I_r$ is now satisfied.
It will turn out to be useful to also add trivial pairs $(|x_j\rangle,|u_j\rangle)=(\mathbf0,\mathbf0)$ to the collection, which do not affect \eqref{eqn:c-cond}. Letting $\tilde{J}$ be the index set over which $j\in \tilde{J}$ runs over and $N_{\tilde{J}}=|\tilde{J}|$, we do this until $N_{\tilde{J}}-r\ge $ the number of columns in $B$. This ensures that $B$ can be padded with zero columns to match the dimensions of $Y$ below, without changing the resulting matrix $BB^T$.

It remains to construct the $C_\alpha$ from the $x_j$ and $u_j$.
Let $V$ be the $N_{\tilde{J}}\times r$ matrix whose $N_{\tilde{J}}$ rows consist of the (real) vectors $u_j^T$. Then $V^TV=I_r$, by the second equation of \eqref{eqn:c-cond}, and so $V$ can be extended to an orthogonal matrix $O=[V\;W]\in\orth(N_{\tilde{J}})$, where $\orth(N_{\tilde{J}})$ is the orthogonal group in dimension $N_{\tilde{J}}$. For $\alpha=1,\ldots,p$, define $D_\alpha:=\operatorname{diag}(x_{\alpha,1},\ldots,x_{\alpha,N_{\tilde{J}}})$.
Using \eqref{eqn:c-cond}, we can check that
\begin{align}\label{eqn:Vconjugate}
V^TD_\alpha V=\mathbf0_r,\quad V^TD_\alpha D_\beta V=H_{\alpha\beta}, \quad O^TD_\alpha O=\begin{pmatrix}\mathbf 0_r&V^TD_\alpha W\\ W^TD_\alpha V&W^TD_\alpha W\end{pmatrix}.
\end{align}
The candidate $O^TD_\alpha O$ is nearly of the form $C_\alpha$, but we need the off-diagonal blocks to be the specific $[B_\alpha\;\mathbf 0]$ and $[B_\alpha\;\mathbf0]^T$. Letting $\wt B_{\alpha}:=V^TD_\alpha W$, and using $WW^T+VV^T=I$ and \eqref{eqn:Vconjugate}, we see that 
\begin{align}
\wt B_\alpha\wt B_\beta^T=V^TD_\alpha(I-VV^T)D_\beta V=H_{\alpha\beta}.
\end{align}
Thus $\wt B_\alpha\wt B_\beta^T=B_\alpha B_\beta^T$ for all $\alpha,\beta$.
We just need to find an orthogonal transformation which maps all the $\wt B_{\alpha}$'s to $B_{\alpha}$'s (padded with zeros). Stack all the $\wt B_{\alpha}'s$ into a big matrix $\wt B$, and stack all the $B_\alpha$'s, padded with zeros to have the same number of columns as the $\wt B_\alpha$'s, into a big matrix $Y$, so that
\begin{align}
\wt B=\begin{pmatrix}\wt B_1\\\vdots\\\wt B_p\end{pmatrix},\quad Y=\begin{pmatrix}[B_1\;\mathbf0]\\\vdots\\ [B_p\;\mathbf0]\end{pmatrix},
\quad \wt B\wt B^T=YY^T.
\end{align}
Then since $\wt B$ and $Y$ have the same dimensions $pr\times (N_{\tilde{J}}-r)$, there is an orthogonal matrix $U\in\orth(N_{\tilde{J}}-r)$ such that $\wt B U=Y$ (e.g. use singular value decomposition) \cite[\S2]{golub2013matrix}. Thus $\wt B_\alpha U=[B_\alpha\;\mathbf0]$ for all $\alpha=1,\ldots,p$.

Now we just replace $W$ with $WU$. Let $O':=[V\;WU]\in\orth(N_{\tilde{J}})$, and define
\begin{align}
C_\alpha:=(O')^TD_\alpha O'=\begin{pmatrix}
V^TD_\alpha V&V^TD_\alpha WU\\
U^TW^TD_\alpha V&U^TW^TD_\alpha WU
\end{pmatrix}
=\begin{pmatrix}
\mathbf0_r&[B_\alpha\;\mathbf0]\\
[B_\alpha\;\mathbf0]^T&K_\alpha
\end{pmatrix},
\end{align}
for symmetric $K_\alpha:=U^TW^TD_\alpha WU$. The $C_\alpha$'s commute since the $D_\alpha$'s are diagonal and $O'$ is the same for all $\alpha=1,\ldots,p$. Thus the $C_\alpha$ give the desired commuting extension.
\end{proof}

\subsection{Simple examples of commuting completion}

We now illustrate the above construction at $p=r=2$.
Take $H=vv^T\in\Sym_4(\R)$ for a unit vector $v\in\R^4$,
and reshape $v$ into the $2\times2$ matrix
$V_{\alpha i}=v_{(\alpha,i)}$.
Then $H$ is real separable if and only if $V$ has rank one. We give two examples below, one where $H$ is separable, the other where it is entangled. The entangled example lies outside the additional restriction
$H\in\Sym_2(\R)\otimes\Sym_2(\R)$ in
Theorem~\ref{thm:sep}. Its obstruction to commuting completion
will therefore be established directly.

As we show in Appendix~\ref{apen:realization}, to realize these matrices as SLD blocks, take
$\mathcal T\cong\C^3$, $\mathcal S\cong\C^2$, and
$\mathcal N\cong\C$, and define
\begin{equation}
T_1=\operatorname{diag}(-1,1,0),\qquad
T_2=\operatorname{diag}(-1,0,1)\;.    
\end{equation}
On $\mathcal H=\mathcal T\oplus\mathcal S\oplus\mathcal N$,
set
\begin{equation}
\rho_0=\frac15(I_5\oplus0),\qquad
L_i^{(0)}=T_i\oplus C_i(0),\qquad
\dot\rho_i=\frac12
\bigl(L_i^{(0)}\rho_0+\rho_0L_i^{(0)}\bigr)\;.    
\end{equation}
These are two-parameter local quantum models with rank-5 states
on a six-dimensional Hilbert space. The SLD representatives have the form
$L_i=T_i\oplus C_i(K_i)$, and hence
\begin{equation}
[L_1,L_2]=0_{\mathcal T}\oplus[C_1(K_1),C_2(K_2)]\;.    
\end{equation}

\subsubsection{Separable case} Let $v=(1,0,0,0)^T$, i.e. $V=\begin{psmallmatrix}1&0\\0&0\end{psmallmatrix}$,
rank one. Then $H=vv^T$ has $\operatorname{Tr}[H]=1$ and admits the trivial rank-one
factorization $B=v$, giving row-blocks $B_1=(1,0)^T$, $B_2=(0,0)^T$. As we discuss in Appendix~\ref{apen:realization}, the $\mathcal{S}\oplus\mathcal{N}$ block of the SLD operators becomes
\begin{equation}
C_1(K_1)=\begin{pmatrix}0&0&1\\0&0&0\\1&0&K_1\end{pmatrix},\qquad
C_2(K_2)=\begin{pmatrix}0&0&0\\0&0&0\\0&0&K_2\end{pmatrix}. 
\end{equation}
A direct computation gives $[C_1(K_1),C_2(K_2)]=\begin{psmallmatrix}0&0&K_2\\0&0&0\\-K_2&0&0\end{psmallmatrix}$,
which vanishes iff $K_2=0$; $K_1$ is completely unconstrained. So a commuting
completion exists (e.g. $K_1=0,\ K_2=0$), consistent with $H$ being separable.

\subsubsection{Entangled case} Let $v=\tfrac1{\sqrt2}(1,0,0,1)^T$, i.e. $V=\tfrac1{\sqrt2}I_2$,
rank two — the direct real analogue of a Bell state, hence entangled. Again
$\operatorname{Tr}[H]=1$ and $B=v$, giving $B_1=\tfrac1{\sqrt2}(1,0)^T$,
$B_2=\tfrac1{\sqrt2}(0,1)^T$ (here we ignore the requirement on $B$ to be rational). Writing $a:=1/\sqrt2$, we have
\begin{equation}
C_1(K_1)=\begin{pmatrix}0&0&a\\0&0&0\\a&0&K_1\end{pmatrix},\qquad
C_2(K_2)=\begin{pmatrix}0&0&0\\0&0&a\\0&a&K_2\end{pmatrix}. 
\end{equation}
Now $[C_1(K_1),C_2(K_2)]=\begin{psmallmatrix}0&a^2&aK_2\\-a^2&0&-aK_1\\-aK_2&aK_1&0\end{psmallmatrix}$,
whose $(1,2)$ entry is $a^2=\tfrac12$, independent of $K_1,K_2$. Therefore, no choice of the
free kernel blocks at this size can make $C_1$ and $C_2$ commute. This obstruction persists under every finite enlargement of the
kernel. Indeed, for $\widehat B_i=[B_i\;\;0]$ and arbitrary
Hermitian kernel blocks $K_i$, the upper-left block of the
commutator is
\begin{equation}
[\widehat C_1,\widehat C_2]_{\mathcal S\mathcal S}
=
\widehat B_1\widehat B_2^\dagger
-\widehat B_2\widehat B_1^\dagger
=
\frac12\begin{pmatrix}0&1\\-1&0\end{pmatrix}
\ne0.    
\end{equation}
Thus no finite commuting completion exists.
By Theorem~\ref{theorem:finitecriterion}, the corresponding
quantum model's QCRB is unattainable.

\section{Proof of quantum model realization (Proposition~\ref{lem:quantum-model})}
\label{apen:realization}
In this appendix we construct the explicit multiparameter estimation model required to prove Proposition~\ref{lem:quantum-model}. The model consists of two main parts: one part encoding the input $H$ to  the \np-hard problem, and an auxiliary commuting part to ensure $\mathcal{J}_\text{Q}\succ0$.
\begin{proof}[Proof of Proposition~\ref{lem:quantum-model}]
We want to construct a family of density matrices $\{\rho_\theta\}$ for $\theta$ near $0$ on some Hilbert space $\mathcal H$, such that its local data at $\theta=0$ satisfies the properties in Proposition~\ref{lem:quantum-model}.
Recall from Appendix~\ref{apen:quantummodel} that we may partition the input $H\in\Sym_p(\R)\otimes\Sym_r(\R)$ into $r\times r$ blocks $(H_{ij})_{i,j=1}^p$. 
For a factorization $H=BB^T$, we can write $B=\left(\begin{smallmatrix}B_1\\\vdots\\B_p\end{smallmatrix}\right)$ with $B_i\in\Q^{r\times q}$, so that $H_{ij}=B_i B_j^T$.

We first define the Hilbert space $\mathcal H$.
Let $\mathcal{S}\cong\C^r$ be the domain of $H_{ij}$ (and $B_i^T$); this space will carry the information about the input $H$. Let $\mathcal N\cong\C^q$ be the domain of $B_i$, so $B_i:\mathcal N\to \mathcal S$.
Let $\mathcal T\cong\C^{p+1}$ be an auxiliary space (its purpose will be to ensure the QFI matrix $\mathcal{J}_\text{Q}$ is positive definite), with an orthonormal basis denoted by $|0\rangle,\ldots,|p\rangle$. Take
\begin{align}\label{eqn:space-decomposition}
\mathcal H=(\mathcal T\oplus \mathcal S)\oplus \mathcal N\cong \C^{p+1+r+q},
\end{align}
and define the density matrix
\begin{align}\label{eq:rho0}
\rho_0:=\frac{1}{r'}I_{\mathcal T\oplus \mathcal S}\oplus\mathbf0_{\mathcal N}=\begin{pmatrix}\frac1{r'}I_{r'}&0\\0&0_q\end{pmatrix},\quad r':=p+1+r,
\end{align}
which has $\supp \rho_0=\mathcal T\oplus \mathcal S$, $\ker \rho_0=\mathcal N$, and $\rank\rho_0=r'$.

On the auxiliary space $\mathcal T$, define the commuting operators
$T_i:=|i\rangle\langle i|-|0\rangle\langle0|$, $1\le i\le p$,
which are Hermitian and traceless. Set
\begin{align}
  L_i^{(0)}&:=
  \begin{pmatrix}
    T_i\oplus0_\mathcal{S}&\binom{0}{B_i}\\[2mm]
    (\,0\ \ B_i^T\,)&0_{\mathcal N}
  \end{pmatrix},
  \label{eq:initialsld}\\
  \rho_i&:=\frac12(L_i^{(0)}\rho_0+\rho_0L_i^{(0)})=\frac1{r'}\begin{pmatrix}
      T_i\oplus 0_\mathcal{S}&\frac12\binom{0}{B_i}\\[2mm]
      (\,0\ \ \frac12 B_i^T\,)&0_{\mathcal N}
  \end{pmatrix}.
  \label{eq:derivatives}
\end{align}
Note all matrices are rational. The commuting operators $T_i$ will be used later to ensure the QFI matrix $\mathcal{J}_\text{Q}$ is positive definite.
We want to show (1) the local data $(\rho_0,\rho_i)$ defined above admits finite commuting SLD representatives if and only if the completion in Theorem~\ref{thm:sep} exists, and (2) the local data comes from a smooth constant-rank model $\{\rho_\theta\}$. The latter we postpone to Appendix~\ref{apen:smoothmodel}.

For (1), we identify the freedom in choosing the SLD operators for this model. 
Let $\mathcal A\cong\C^d$ be any finite-dimensional auxiliary kernel space, and consider the extension $\mathcal H\oplus\mathcal A$, and the zero extensions $\rho_0\mapsto\rho_0\oplus 0_{\mathcal A}$, $\rho_i\mapsto\rho_i\oplus 0_{\mathcal A}$.
Regardless of how large the dimension of the extended kernel space $\mathcal N\oplus \mathcal A$ is, the extended state is $\rho_0=r'^{-1}\operatorname{diag}(I_{r'},0)$ in the support-kernel decomposition $(\mathcal T\oplus \mathcal S)\oplus(\mathcal N\oplus \mathcal A)$. For a Hermitian SLD candidate
$L=\begin{psmallmatrix}X&Y\\Y^\dagger&Z\end{psmallmatrix}$, written in the support-kernel decomposition, we see
\begin{equation}
      \frac12(L\rho_0+\rho_0L)
  =\frac1{2r'}
  \begin{pmatrix}2X&Y\\Y^\dagger&0\end{pmatrix}.
\end{equation}
Comparing with the fixed derivatives $\rho_i\oplus0_{\mathcal A}$ defined through \eqref{eq:derivatives} fixes the support block and support--kernel blocks of $L$, while leaving the kernel--kernel block $Z$ arbitrary.  
Then writing in the decomposition $\mathcal T\oplus(\mathcal S\oplus\mathcal N\oplus \mathcal A)$, we see every SLD operator on this extension is of the form
\begin{equation}
 L_i= T_i\oplus
  \begin{pmatrix}
    0_{\mathcal S}&[\,B_i\ 0\,]\\
    [\,B_i\ 0\,]^\dagger&K_i
  \end{pmatrix}\;,
  \label{eq:allslds}
\end{equation}
for Hermitian $K_i$. Thus, we can write the SLDs as $L_i=T_i\oplus C_i(K_i)$, where $C_i(K_i):=  \begin{pmatrix}
    0_{\mathcal S}&[\,B_i\ 0\,]\\
    [\,B_i\ 0\,]^\dagger&K_i
  \end{pmatrix}$ is of exactly the same form as the extensions used in Theorem~\ref{thm:sep}. Finally, recall the $T_i$'s commute, and observe that
\begin{equation}
    [L_i,L_j]=[T_i,T_j]\oplus[C_i(K_i),C_j(K_j)]=0\oplus[C_i(K_i),C_j(K_j)]\;.
\end{equation}
Therefore, the existence of commuting SLD operators for this problem is exactly equivalent to the existence of a commuting completion in Theorem~\ref{thm:sep}.

Finally, we wish to show that the QFI matrix $\mathcal{J}_\text{Q}$ for this model is positive definite. This is the purpose of the operators $T_i$ and auxiliary space $\mathcal T$. We have 
\begin{equation}
\label{eq:qfiTH}
    \mathcal{J}_{\text{Q},{ij}}=\frac12\Tr[\rho_0(L_iL_j+L_jL_i)]=\frac{1}{r'}\left(\Tr[T_iT_j]+\Tr\left[B_iB_j^T\right]\right)\;.
\end{equation}
By construction $\Tr[T_iT_j]=1+\delta_{ij}$, which contributes a term $I_p+\mathbf1\mathbf1^T$, for $\mathbf1=(1,\ldots,1)^T\in\R^p$. 
This is positive definite, as for any nonzero $c\in\mathbb{R}^p$
\begin{equation}
    c^T(I_p+\mathbf{1}\mathbf{1}^T)c=||c||_2^2+(\mathbf{1}^Tc)^2>0\;.
\end{equation}
Denote the second term as the matrix $G$ such that $G_{ij}=\Tr[B_iB_j^T]$. Observe that $\Tr[B_iB_j^T]$ is the Frobenius inner product. For any $c\in\mathbb{R}^p$ we have
\begin{equation}
    \begin{split}
        c^TGc&=\Tr\left[\sum_ic_iB_i\sum_jc_jB_j^T\right]
    =\left|\left|\sum_ic_iB_i\right|\right|_F^2
    \geq0\;.
    \end{split}
\end{equation}
Hence $G\succeq0$. Therefore, the QFI matrix is
\begin{equation}
\label{eq:QFIfineq}
    \mathcal{J}_\text{Q}=\frac{1}{r'}(I_p+\mathbf{1}\mathbf{1}^T+G)\succ0\;.
\end{equation}

The dimension and bit length are polynomial in those of the factorization $BB^T$.  Equation~\eqref{eq:rho0} makes $\rho_0$ maximally mixed on its support, the prescribed support--support blocks $T_i\oplus0_{\mathcal S}$ commute, and the source tensor enters only through the support--kernel blocks $B_i$.  This proves every assertion of Proposition~\ref{lem:quantum-model}. Furthermore, these properties hold for the smooth model discussed below. 
\end{proof}

\subsection{Smooth model}
\label{apen:smoothmodel}

We now construct a smooth constant-rank model that gives rise to the above derivatives in \eqref{eq:derivatives}. Note that constant-rank here means that the rank does not change with $\theta$. Set
\begin{equation}
X(\theta)=\sum_i\theta_iL_i^{(0)}-\psi(\theta)I,
  \qquad
  \rho_\theta=
\mathrm{e}^{\frac12X(\theta)}\rho_0\mathrm{e}^{\frac12X(\theta)}\;,
  \label{eq:smoothmodel}
\end{equation}
where $\psi(\theta)=\log(\Tr[\rho_0 \mathrm{e}^{\sum_i\theta_iL_i^{(0)}}])$ is a normalisation factor. Note that
\begin{equation}
    \frac{\partial}{\partial\theta_i}(\mathrm{e}^{\frac12X(\theta)}\rho_0\mathrm{e}^{\frac12X(\theta)})\big|_{\theta=0}=\frac12(\rho_0 L_i^{(0)}+L_i^{(0)}\rho_0)\;.
\end{equation}
Therefore, $\partial\rho_\theta/\partial\theta_i|_{\theta=0}=\rho_i$. As $\mathrm{e}^{X(\theta)/2}$ is invertible the rank of $\rho_\theta$ is the same as the rank of $\rho_0$.

\section{Approximate attainability}
\label{apen:approxattain}
In this appendix, we consider an approximate version of the attainability argument studied in the main text (Eq.~\eqref{eq:polypr}). Specifically, we consider the promise problem of deciding whether $V-\mathcal{J}_\text{Q}^{-1}=0$ or $\Tr[V-\mathcal{J}_\text{Q}^{-1}]\geq\epsilon$:

\begin{definition}[Gap attainability on the encoded model]
\label{def:encoded-gap-attainability}
Let $\mathcal M_H$ be the model constructed in
Proposition~\ref{lem:quantum-model} from a rational
$H\succeq0$ with $\Tr [H]=1$, and define
\begin{equation}
  \Delta(H)
  :=
  \inf_{(\Pi,\hat{\theta})}
  \Tr\!\bigl[V(\Pi,\hat{\theta})-\mathcal{J}_\text{Q}^{-1}\bigr]
  =
  \mathcal{C}_{\mathrm{MI}}(\mathcal M_H)-\mathcal{C}_{\mathrm Q}(\mathcal M_H).
\end{equation}
Given a rational $0<\varepsilon\le1$, distinguish the promised cases
\begin{equation}
  \mathrm{YES}:\quad V(\Pi,\hat{\theta})=\mathcal{J}_\text{Q}^{-1},
  \qquad
  \mathrm{NO}:\quad \Delta(H)\ge\varepsilon.
\end{equation}
\end{definition}

Note that the YES instance immediately implies $\Delta(H)=0$, and that since this is a promise problem, we do not require a correct output for $0<\Delta(H)<\varepsilon$. We will show that this problem is \np-hard and belongs to $\sans{PromiseNP}$\   ($\sans{PromiseNP}$ is the class of promise problems whose YES instances admit polynomial-size certificates verifiable in polynomial time\footnote{Strictly speaking, the problems considered here, as well as the weak-membership problems used in Appendix~\ref{apen:Gurvits}, are promise problems. The class $\sans{PromiseNP}$ is the promise-problem analogue of $\np$: YES instances admit polynomial-size certificates verifiable in polynomial time, and no certificate is accepted for NO instances, while no condition is imposed on inputs outside the promise.}).
Since the promise problem in Definition~\ref{def:encoded-gap-attainability} is an easier problem than the decision problem to decide $\Delta(H)<\varepsilon$ vs $\Delta(H)\ge\varepsilon$, we will also obtain \np-hardness of the latter.

\subsection{\np-hardness of approximate attainability}\label{subsec:approx-np}
We first prove that the problem in Definition~\ref{def:encoded-gap-attainability} is \np-hard. The reduction is from Gharibian's result on weak membership with an inverse-polynomial tolerance \cite{gharibian2010strong}.

\begin{proposition}[Approximate attainability]
\label{prop:approximate-attainability}f
Consider the setting of Definition~\ref{def:encoded-gap-attainability}.
For every $0<\delta\le1$,
\begin{equation}
  \dist(H,\mathcal K^{\R}_{p,r})\ge\delta
  \quad\Longrightarrow\quad
  \Delta(H)\ge
  \frac{\delta^2}
       {100(p+r+1)(pr+1)}.
  \label{eq:approximate-attainability-gap}
\end{equation}
Consequently, it is \np-hard to distinguish exact attainability ($V=\mathcal{J}_\text{Q}^{-1}\implies\Delta(H)=0$) from the
case in which every locally unbiased estimator satisfies
\begin{equation}
  \Tr[V-\mathcal{J}_\text{Q}^{-1}]\ge\varepsilon
\end{equation}
for $\varepsilon$ inverse-polynomial in $r$ and $p$.
\end{proposition}

\begin{proof}
Put $r'=p+r+1$. Recall the model $\mathcal M_H$ from the proof of Proposition~\ref{lem:quantum-model} is defined on the Hilbert space $\mathcal H=(\mathcal T\oplus\mathcal S)\oplus\mathcal N\cong \C^{p+1+r+q}$ [Eq.~\eqref{eqn:space-decomposition}], with $\mathcal S$ encoding the information about $H$, and with $\rho_0:=\frac{1}{r'}I_{\mathcal T\oplus \mathcal S}\oplus\mathbf0_{\mathcal N}$.
We prove Eq.~\eqref{eq:approximate-attainability-gap} by showing that if $\Delta(H)<\delta^2/(100r'(pr+1))$, then $\dist(H,\mathcal K^{\R}_{p,r})<\delta$. 
If $\Delta(H)<\delta^2/(100r'(pr+1))$, there is some locally unbiased estimator $(\Pi,\hat{\theta})$ such that
\begin{equation}
     Q=V-\mathcal{J}_\text{Q}^{-1},
  \qquad
  \eta=\Tr [Q]<\frac{\delta^2}{100r'(pr+1)}\;.
\end{equation}
Note that $Q\succeq0$. Intuitively, $\eta$ measures how far the estimator is from the QCRB. By Proposition~\ref{prop:finite-support}, we can restrict to a finite outcome POVM. After a finite Naimark
dilation, any POVM can be written as a projective measurement $\{P_a\}$. Therefore, the matrices $(Z_j=\sum_a\hat{\theta}_j(a)P_a)$
can be written as commuting Hermitian matrices
$Z_1,\ldots,Z_p$ satisfying
\begin{equation}
      \langle Z_i,L_j\rangle_{\rho_0}=\delta_{ij},
  \qquad
  \langle Z_i,Z_j\rangle_{\rho_0}=V_{ij},\qquad \langle L_i,L_j\rangle_{\rho_0}=\mathcal{J}_{\text{Q},ij}
\end{equation}
where
\begin{equation}
      \langle X,Y\rangle_{\rho_0}
  :=
  \frac12\Tr\!\bigl[\rho_0(XY+YX)\bigr].
\end{equation}
Let
\begin{equation}
   L^i=\sum_j(\mathcal{J}_\text{Q}^{-1})_{ij}L_j,
  \qquad
  A_i=\sum_j\mathcal{J}_{\text{Q},ij}Z_j\;,
\end{equation}
so that $\langle L^i,L_j\rangle_{\rho_0}=\delta_{ij}$ and $\langle L^i,L^j\rangle_{\rho_0}=(\mathcal{J}_{\text{Q}}^{-1})_{ij}$. The matrices $A_i$ commute, and
\begin{equation}
  \langle Z_i-L^i,Z_j-L^j\rangle_{\rho_0}=V_{ij}-(\mathcal{J}_{\text{Q}}^{-1})_{ij}=Q_{ij}\;.    
\end{equation}

Let $P$ be the projection of $\rho_0$ onto its support $\mathcal T\oplus\mathcal S$. Since $\rho_0=P/r'$, for any Hermitian matrix $B$ we have $\langle B,B\rangle_{\rho_0}=||BP||_{\mathrm F}^2/r'$, where the norm $\|\cdot\|_{\mathrm{F}}$ is the Frobenius norm. Therefore, we can write
\begin{equation}
     \begin{split}
             \sum_i\norm{(A_i-L_i)P}_{\mathrm F}^2
  &=r'\sum_i \langle A_i-L_i,A_i-L_i\rangle_{\rho_0}\\
    &=r'\sum_{i,j,k} \langle \mathcal{J}_{\text{Q},ij}Z_j-\mathcal{J}_{\text{Q},ij}L^j,\mathcal{J}_{\text{Q},ik}Z_k-\mathcal{J}_{\text{Q},ik}L^k\rangle_{\rho_0}\\
    &=r'\sum_{i,j,k} \mathcal{J}_{\text{Q},ij}\mathcal{J}_{\text{Q},ik}\langle Z_j-L^j,Z_k-L^k\rangle_{\rho_0}\\ &=r'\sum_{i,j,k} \mathcal{J}_{\text{Q},ij}Q_{jk}\mathcal{J}_{\text{Q},ik}\\
  &=r'\sum_{i,j,k}
\mathcal{J}_{\mathrm Q,ij}Q_{jk}\mathcal{J}_{\mathrm Q,ki}\\&=r'\Tr[\mathcal{J}_\text{Q}Q\mathcal{J}_\text{Q}].
     \end{split}
\end{equation}
Equation~\eqref{eq:QFIfineq} gives
\begin{equation}
  \mathcal{J}_\text{Q}=\frac1{r'}\bigl(I+\one\one^T+G\bigr),
  \qquad
  G\succeq0,
  \qquad
  \Tr [G]=\Tr [H]=1.
\end{equation}
Hence $\norm{\mathcal{J}_\text{Q}}_{\mathrm{op}}\le1$ ($||A||_{\text{op}}$ is the operator norm), and therefore
\begin{equation}
  \sum_i\norm{(A_i-L_i)P}_{\mathrm F}^2
  \le r'\norm{\mathcal{J}_\text{Q}}_{\mathrm{op}}^2\Tr[Q]\le r'\eta.
  \label{eq:approximate-sld-bound}
\end{equation}
This tells us that if the covariance is close to \(\mathcal{J}_\text{Q}^{-1}\), then the commuting \(A_i\)'s are jointly close to the SLD operators \(L_i\).
Since the commuting $A_i$'s will produce a separable tensor $\widetilde H$,  their closeness to the $L_i$'s will imply that $H$ must also be close to the separable $\widetilde H$.

Let $R:\mathcal H\to\mathcal S$ be the projection onto the subspace $\mathcal S$ which encodes the model $H$, and set
\begin{equation}
     X=[\,A_1R\ \cdots\ A_pR\,],
  \qquad
  Y=[\,L_1R\ \cdots\ L_pR\,]\;.
\end{equation}
Then the $(i,j)$ block of $Y^*Y$ is $H_{ij}$ using Eq.~\eqref{eq:initialsld} and $H_{ij}=B_iB_j^T$, and so we see 
\begin{equation}
      Y^*Y=H,
  \qquad
  \norm{Y}_{\mathrm F}^2=\Tr H=1\;.
\end{equation}
Moreover, \eqref{eq:approximate-sld-bound} implies
\begin{equation}
    \tau:=\norm{X-Y}_{\mathrm F}\le\sqrt{r'\eta}\;.  
\end{equation}

Since the $A_i$ commute, the $(\Leftarrow)$ direction of the proof
of Theorem~\ref{thm:sep} shows that
\begin{equation}
      \widetilde H:=\Re(X^*X)\in\Sep^{\R}_{p,r}.
\end{equation}
While $\widetilde H$ is not normalized to have trace 1, we can normalize it to obtain $\widetilde H_1:=\widetilde H/\Tr\widetilde H\in\mathcal K_{p,r}^\R$. 
Since $\widetilde H$ is close to the trace-1 $H$, we will show that $\Tr[\widetilde H]$ is close to 1 and so $H$ is close to $\widetilde H_1\in\mathcal K_{p,r}^\R$.
Let $t=\Tr[\widetilde H]$, so that $|t-1|=\left|\Tr [\widetilde H-H]\right|$. For an $n\times n$ matrix $M$ we have $|\Tr[M]|\le\sqrt{n}\norm{M}_\text{F}$. As $H$ is a $pr\times pr$ matrix, then
\begin{equation}\label{eqn:trace-dist-1}
      |t-1|
  \le
  \sqrt{pr}\|\widetilde H-H\|_{\mathrm F},
\end{equation}
and we can estimate
\begin{equation}
\begin{split}
     s:=\norm{\widetilde H-H}_{\mathrm F}
  &\le
  \norm{X^*X-Y^*Y}_{\mathrm F}\\
  &=\norm{X^*(X-Y)+(X^*-Y^*)Y}_{\mathrm F}\\
  &\le (\norm{X}_{\mathrm F}+\norm{Y}_{\mathrm F})\norm{X-Y}_{\mathrm F}\\&\le
  (2+\tau)\tau\;,
\end{split}
  \label{eq:approximate-gram-bound}
\end{equation}
where in the last line we have used $\norm{X}_{\mathrm F}\le\norm{Y}_{\mathrm F}+\norm{X-Y}_{\mathrm F}=1+\tau$.

Recall that
\begin{equation}
      \eta<
  \frac{\delta^2}{100r'(pr+1)}\;.
\end{equation}
Then
\begin{equation}
   \tau<\frac{\delta}{10\sqrt{pr+1}}\le\frac1{10},\qquad  s<3\tau\;, 
\end{equation}
as $\delta\le1$.
From Eq.~\eqref{eqn:trace-dist-1}, we see
\begin{equation}
      |t-1|
  \le
  \sqrt{pr}\,s\le3\sqrt{pr}\tau\le \frac{3\delta\sqrt{pr}}{10\sqrt{pr+1}}
  <
  \frac3{10},
\end{equation}
so $t>1/2$. 
Using $\norm{H}_{\mathrm F}\le1$, we obtain
\begin{equation}
\begin{split}
     \dist(H,\mathcal K^{\R}_{p,r})
  &\le
  \norm{\widetilde H/t-H}_{\mathrm F}    \\
  &=\norm{\frac{\widetilde H-H}{t}+\left(\frac{1}{t}-1\right)H}_{\mathrm F} \\
  &\le\frac{s}{t}+\frac{|1-t|}{t}\norm{H}_{\mathrm F}\\
  &\le 2s+2|1-t|
  \\&\le
  2(1+\sqrt{pr})s  \\
  &<6(1+\sqrt{pr})\tau   \\
  &<\delta,
\end{split}
\end{equation}
where the last inequality follows from $6(1+\sqrt{pr})<10\sqrt{pr+1}$. This proves \eqref{eq:approximate-attainability-gap}.

Gharibian proves that weak membership for separability is \np-hard for
an inverse-polynomial tolerance $\beta$ \cite{gharibian2010strong}. 
We now show that the inverse-polynomial weak-membership hardness of
Ref.~\cite{gharibian2010strong} also holds for
$\mathcal K^{\mathbb R}_{p,r}$. To see this, put
$\mathcal W=\operatorname{Sym}_p(\mathbb R)\otimes
\operatorname{Sym}_r(\mathbb R)$.
The projection onto $\mathcal W$, $\Pi(X)=\tfrac14(X+X^{T_A}+X^{T_B}+X^T),$ maps the set of complex separable states onto the real separable set
$\mathcal K^{\mathbb R}_{p,r}$. On product states it acts as $\Pi(zz^\dagger\otimes ww^\dagger)
=\operatorname{Re}(zz^\dagger)\otimes
 \operatorname{Re}(ww^\dagger)$, whose factors are real positive semidefinite matrices. Consequently, the set of real separable states is the intersection of the set of complex separable states with $\mathcal W$, and every objective in $\mathcal W$ has the same maximum over both bodies i.e. for $C\in\mathcal W$, $\Tr[CX]=\Tr[\Pi(C)X]=\Tr[C\Pi(X)]$. Hence, for such objectives, optimisation over $X$ can be restricted to optimisation over $X\in\mathcal W$. Gharibian's hard objectives have real symmetric blocks and belong to $\mathcal W$, so their optimal values are unchanged by this restriction and admit real product optimizers.

Moreover, every trace-one matrix in $\mathcal W$ within Frobenius
distance $1/\sqrt{pr(pr-1)}$ of $I_{pr}/pr$ is complex separable (see Theorem 6.4 of~\cite{gurvits2004classical}),
and hence real separable by the characterization above. Gharibian's optimization-to-membership reduction is formulated
in Euclidean coordinates. We use the independent entries of a
trace-one matrix in $\mathcal W$ as rational coordinates. If $x,y$
represent matrices $X,Y$, then
\begin{equation}
    \|x-y\|_2\leq\|X-Y\|_{\mathrm F}\leq\sqrt{\max\{pr,4\}}\,\|x-y\|_2\;.
\end{equation}
Thus, the ball used in Gharibian's work still has inverse-polynomial radius in these coordinates, and the reduction gives inverse-polynomial weak-membership hardness for $\mathcal K^{\mathbb R}_{p,r}$.

Given a promised weak-membership instance $(H,\beta)$, first answer
NO directly if $H\not\succeq0$.  Otherwise construct $\mathcal M_H$ and
set
\begin{equation}
  \varepsilon
  =
  \frac{\beta^2}
       {100(p+r+1)(pr+1)}.
\end{equation}
Recalling Eq.~\eqref{eq:defweaksep}, in the YES case,
$H\in S(\mathcal K^{\R}_{p,r},-\beta)
  \subseteq\mathcal K^{\R}_{p,r}$,
so Theorems~\ref{theorem:finitecriterion} and
\ref{th:quantum-encoding} give exact attainability.  In the NO case,
$H\notin S(\mathcal K^{\R}_{p,r},\beta)$,
and hence $\dist(H,\mathcal K^{\R}_{p,r})>\beta$.  Equation
\eqref{eq:approximate-attainability-gap} then gives
\begin{equation}
  \Tr[V-\mathcal{J}_\text{Q}^{-1}]\ge\varepsilon
\end{equation}
for every locally unbiased estimator.  Since $\beta$ is inverse
polynomial, so is $\varepsilon$, and the construction is polynomial.
\end{proof}

\subsection{Polynomial certificates for approximate attainability}
\label{subsec:approx-np-complete}
Having established $\np$-hardness, we now show that the problem in Definition~\ref{def:encoded-gap-attainability} belongs to $\sans{PromiseNP}$. We will show that the gap allowed in Definition~\ref{def:encoded-gap-attainability} can be used to give a polynomial-time certificate for approximate separability of $H$. Essentially, the gap will allow replacing a general (non-rational) separable decomposition of $H$ with a rational separable approximation, which can then serve as the polynomial-size certificate for $\Delta(H)=0$.

First, we prove stability in the other direction compared to Eq.~\eqref{eq:approximate-attainability-gap}. This will later be used to show there is no such certificate for $\Delta(H)\ge\varepsilon$.

\begin{lemma}
\label{lem:upper-stability}
Consider the setting of Definition~\ref{def:encoded-gap-attainability}, and put $r'=p+r+1$. If
\begin{equation}
  \dist(H,\mathcal K^{\R}_{p,r})\le\xi
  \qquad\text{and}\qquad
\xi\le\frac{1}{4r'\sqrt{pr}},
\end{equation}
then
\begin{equation}
  \Delta(H)\le4r'\sqrt{pr}\,\xi.
\end{equation}
\end{lemma}

\begin{proof}
Choose $\widehat H\in\mathcal K^{\R}_{p,r}$ such that
\begin{equation}
      \norm{H-\widehat H}_{\mathrm F}\le\xi\;.
\end{equation}
Let $H=BB^T$ be the rational factorization used in the construction of
$\mathcal M_H$. Note that this is not unique: we can pad $B$ such that $B\in\mathbb{R}^{pr\times Q}$ where $Q\geq\max(pr,q)$. Then there exists $U\in\mathbb{R}^{pr\times Q}$ such that $UU^T=I_{pr}$ and $B=H^{1/2}U$. After adjoining zero columns, one may choose a 
factorization $\widehat H=\widehat B\widehat B^T$ where $\widehat B=\widehat H^{1/2}U$ such that
\begin{equation}
  b^2
  :=
  \norm{B-\widehat B}_{\mathrm F}^2
  \le
  \norm{H-\widehat H}_1
  \le
  \sqrt{pr}\,\xi\;,
  \label{eq:aligned-gram-bound}
\end{equation}
where we use the fact that $\norm{H^{1/2}-\widehat H^{1/2}}_\text{F}^2\le \norm{H-\widehat H}_1$.

Since $\widehat H$ is real separable, Theorem~\ref{thm:sep} gives
commuting Hermitian matrices $\tilde{A}_1,\ldots,\tilde{A}_p$ whose prescribed
support--kernel blocks are $\widehat B_1,\ldots,\widehat B_p$. The support block of each $\tilde{A}_i$ is $T_i$. Let
$L_i$ be the SLDs of the model $\mathcal M_H$ and put
\begin{equation}
      E_i=\tilde{A}_i-L_i.
\end{equation}
$\tilde{A}_i$ and $L_i$ have identical support--support blocks corresponding to $\mathcal T\oplus\mathcal S$, and therefore
\begin{equation}
     \sum_i\norm{E_iP}_{\mathrm F}^2=b^2.
\end{equation}
Since $\rho_0=P/r'$,
\begin{equation}
     \alpha^2
  :=
  \sum_i\langle E_i,E_i\rangle_{\rho_0}
  =
  \frac{b^2}{r'}.
\end{equation}
Define
$\Gamma_{ij}=\langle E_i,L_j\rangle_{\rho_0}$ and $\Sigma_{ij}=\langle E_i,E_j\rangle_{\rho_0}$.
For the encoded model, $\frac1{r'} I\preceq \mathcal{J}_\text{Q}\preceq I$ follows from Eq.~\eqref{eq:QFIfineq}. For unit vectors $x,y\in\R^p$, Cauchy--Schwarz gives
\begin{equation}
    \begin{split}
        |x^T\Gamma y|
&=
\left|
\left\langle
\sum_i x_iE_i,\,
\sum_j y_jL_j
\right\rangle_{\rho_0}
\right| \\
&\le\sqrt{\left\langle\sum_i x_iE_i,\sum_k x_kE_k\right\rangle_{\rho_0}\left\langle
\sum_j y_jL_j,\sum_l y_lL_l\right\rangle_{\rho_0}}
\le
\alpha\,\sqrt{y^T\mathcal{J}_\text{Q}y}
\le\alpha\;.
    \end{split}
\end{equation}
 Hence $\|\Gamma\|_{\mathrm{op}}\le\alpha$. Moreover, $\mathcal{J}_\text{Q}\succeq (r')^{-1}I$ implies
\begin{equation}
    \|\mathcal{J}_\text{Q}^{-1}\|_{\mathrm{op}}
=
\lambda_{\max}(\mathcal{J}_\text{Q}^{-1})
=
\frac{1}{\lambda_{\min}(\mathcal{J}_\text{Q})}
\le r'.
\end{equation}

By \eqref{eq:aligned-gram-bound} and the assumption on $\xi$,
$r'\alpha=\sqrt{ r'}\,b\le\frac12.$
Hence $G:=\mathcal{J}_\text{Q}+\Gamma=\mathcal{J}_\text{Q}(I+\mathcal{J}_\text{Q}^{-1}\Gamma)$ is invertible and
\begin{equation}
  \norm{G^{-1}}_{\mathrm{op}}\le2r'.
  \label{eq:G-inverse-bound}
\end{equation}

As the $\tilde{A}_i$ commute, they can be measured via some projective measurement. If $a=(a_1,\ldots,a_p)^T$ is the joint eigenvalue vector, use
the estimator $\hat{\theta}=G^{-1}a$. Since $\Tr[\rho_0\tilde{A}_i]=0$ (as $\tilde{A}_i$ and $L_i$ have identical support-support blocks corresponding to $\mathcal{T}\oplus\mathcal{S}$) and
\begin{equation}
  G_{ij}=\langle E_i+L_i,L_j\rangle_{\rho_0}
  =
  \langle \tilde{A}_i,L_j\rangle_{\rho_0}
  =
  \Tr[\rho_j\tilde{A}_i],
\end{equation}
this estimator is locally unbiased (as $G^{-1}G=I$).

Let $K_{ij} =\langle \tilde{A}_i,\tilde{A}_j\rangle_{\rho_0}.$
The covariance of this measurement is
$V=G^{-1}KG^{-T}.$ Since $\tilde{A}_i=L_i+E_i$,
\begin{equation}
K=\mathcal{J}_\text{Q}+\Gamma+\Gamma^T+\Sigma.
\end{equation}
Using $G\mathcal{J}_\text{Q}^{-1}G^T=\mathcal{J}_\text{Q}+\Gamma+\Gamma^T+\Gamma\mathcal{J}_\text{Q}^{-1}\Gamma^T$, it follows that
\begin{equation}
  K-G\mathcal{J}_\text{Q}^{-1}G^T
  =
  \Sigma-\Gamma \mathcal{J}_\text{Q}^{-1}\Gamma^T
  \succeq0.
\end{equation}
Positive-semidefiniteness follows as
$
    \begin{pmatrix}
        \Sigma&\Gamma\\
        \Gamma^T&\mathcal{J}_\text{Q}
    \end{pmatrix}\succeq0.
$
Therefore
\begin{equation}
  V-\mathcal{J}_\text{Q}^{-1}
  =
  G^{-1}
  \bigl(\Sigma-\Gamma \mathcal{J}_\text{Q}^{-1}\Gamma^T\bigr)
  G^{-T}.
\end{equation}
Using \eqref{eq:G-inverse-bound},
\begin{equation}
\begin{aligned}
  \Tr[V-\mathcal{J}_\text{Q}^{-1}]
  &\le
  \norm{G^{-1}}_{\mathrm{op}}^2\Tr[\Sigma]\\
  &\le
  4(r')^2\alpha^2\\
  &=
  4r'b^2\\
  &\le
  4r'\sqrt{pr}\,\xi.
\end{aligned}
\end{equation}
Taking the infimum over locally unbiased estimators proves the result.
\end{proof}

Proposition~\ref{prop:approximate-attainability} shows that $\dist(H,\mathcal K^{\R}_{p,r})\ge\delta$ implies $\Delta(H)\ge\frac{\delta^2} {100(p+r+1)(pr+1)}$. Conversely, Lemma~\ref{lem:upper-stability} shows that $ \dist(H,\mathcal K^{\R}_{p,r})\le\xi$
with $\xi\le1/(4r'\sqrt{pr})$ implies $\Delta(H)\le4r'\sqrt{pr}\,\xi.$ Proposition~\ref{prop:approximate-attainability} shows that the approximate attainability problem is \np-hard. We now wish to show that it is also $\sans{PromiseNP}$. For this we require that every YES instance has a polynomial-size certificate that a deterministic polynomial-time verifier accepts, while NO instances do not have such a certificate.

\begin{theorem}[Complexity of approximate attainability on the encoded model]
\label{thm:gap-promise-np-complete}
The gap attainability problem of
Definition~\ref{def:encoded-gap-attainability} is \np-hard under polynomial-time
Turing reductions, and belongs to
$\sans{PromiseNP}$.
\end{theorem}

\begin{proof}
Proposition~\ref{prop:approximate-attainability} proves \np-hardness under polynomial-time Turing reductions. It remains to prove membership in
$\sans{PromiseNP}$.

Let
\begin{equation}
  m
  =
  \dim\bigl(\Sym_p(\R)\otimes\Sym_r(\R)\bigr)-1
  =
  \frac{p(p+1)r(r+1)}4-1,
\end{equation}
and set
\begin{equation}
  \xi_0=\frac{\varepsilon}{8(p+r+1)pr}.
\end{equation}
A certificate consists of at most $m+1$ rational weights $q_a\ge0$
and nonzero rational vectors $z_a\in\Q^p$ and $w_a\in\Q^r$, satisfying
$\sum_aq_a=1.$
Given the certificate, the verifier forms
\begin{equation}
  \widehat H
  =
  \sum_aq_a
  \frac{z_az_a^T}{z_a^Tz_a}
  \otimes
  \frac{w_aw_a^T}{w_a^Tw_a}.
\end{equation}
Note that this construction avoids square roots and ensures that everything remains rational. Thus, $\widehat H\in\mathcal K^{\R}_{p,r}$ exactly. The verifier accepts
if
\begin{equation}
  \norm{H-\widehat H}_{\mathrm F}\le\xi_0.
\end{equation}

Suppose first that the instance is a YES instance.
Theorems~\ref{theorem:finitecriterion} and \ref{th:quantum-encoding} imply $ H\in\mathcal K^{\R}_{p,r}$.
By Carath\'eodory's convex hull theorem, $H$ is a convex combination of at most
$m+1$ real pure product states. Approximating the weights and product
vectors by rationals gives a certificate satisfying
\begin{equation}
  \norm{H-\widehat H}_{\mathrm F}\le\xi_0.
\end{equation}
Since $\xi_0^{-1}$ has polynomial encoding length, the certificate
has polynomial size.

Conversely, suppose that the verifier accepts. By
Lemma~\ref{lem:upper-stability},
\begin{equation}
\begin{aligned}
  \Delta(H)
  &\le
  4(p+r+1)\sqrt{pr}\,\xi_0\\
  &=
  \frac{\varepsilon}{2\sqrt{pr}}\\
  &\le
  \frac{\varepsilon}{2}
  <
  \varepsilon.
\end{aligned}
\end{equation}
Hence no instance satisfying the NO promise
$\Delta(H)\ge\varepsilon$ can possess such an accepting certificate.
Therefore the gap problem belongs to $\sans{PromiseNP}$.
\end{proof}

\section{Attainability of the QCRB as quadratic equations}
\label{apen:existence}
In this appendix we rewrite the problem of deciding the exact attainability of the QCRB as a set of quadratic equations. Write each SLD operator as
\begin{equation}
    L_i=\begin{pmatrix}
        L_{\text{S},i}&L_{\text{SK},i}\\
        L_{\text{SK},i}^\dagger &L_{\text{K},i}
    \end{pmatrix}\;,
\end{equation}
where $L_{\text{K},i}$ can be chosen as any Hermitian matrix. Consider a simple two-parameter estimation problem. We know from Theorem~\ref{theorem:finitecriterion} that the QCRB is attainable if and only if one can choose $L_{\text{K},1},L_{\text{K},2}$ such that
\begin{equation}
\begin{split}
\label{eq:apencom}
&L_{\text{S},1}L_{\text{S},2}-L_{\text{S},2}L_{\text{S},1}+L_{\text{SK},1}L_{\text{SK},2}^\dagger-L_{\text{SK},2}L_{\text{SK},1}^\dagger=0\\
&L_{\text{S},1}L_{\text{SK},2}-L_{\text{S},2}L_{\text{SK},1}+L_{\text{SK},1}L_{\text{K},2}-L_{\text{SK},2}L_{\text{K},1}=0\\
&L_{\text{SK},1}^\dagger L_{\text{SK},2}-L_{\text{SK},2}^\dagger L_{\text{SK},1}+[L_{\text{K},1},L_{\text{K},2}]=0
\end{split}
\end{equation}
As there are no free terms in the first equation, this offers a simple check that can reveal when the QCRB is not attainable. Going forward, we assume the first equation is satisfied. Writing the independent real parameters of the Hermitian kernel blocks \(L_{\mathrm K,i}\) as a vector \(x\), the commutativity conditions are therefore equivalent, for any fixed extension dimension, to a finite system of polynomial equations of degree at most two:
\begin{equation}
\label{eq:quadratic}
\begin{split}
&b^{(m)}.x+c_l^{(m)}=0\\
&x.A^{(m)}.x+c_q^{(m)}=0\;,\\
\end{split}
\end{equation}
where $x$ is a vector encoding the free elements of all $L_{\text{K},i}$, $b^{(m)}$ is a vector, $A^{(m)}$ is a square matrix, and $c_l^{(m)}$ and $c_q^{(m)}$ are constants. $b^{(m)}$, $A^{(m)}$, $c_l^{(m)}$ and $c_q^{(m)}$ are completely determined from the problem. The commutativity conditions in Eq.~\eqref{eq:apencom} are equivalent to checking for the existence of a solution $x$ to the equations, Eq.~\eqref{eq:quadratic}. Thus, for a fixed extended kernel dimension, QCRB attainability has the form of a structured feasibility problem over the reals, closely related to problems studied in the existential theory of the reals; see, for example, problem (A1) of Ref.~\cite{schaefer2026existential}. However, we do not claim that the attainability problem is as hard as this complexity class. This is not proven here, as the form of the commutator in Eq.~\eqref{eq:apencom} does not allow for arbitrary quadratic equations in Eq.~\eqref{eq:quadratic}.

This perspective does however provide crucial insight into classes of problems that are easily decidable. For fixed kernel dimensions, for certain problems it is possible that the free vector $x$ may be completely constrained by the linear equations in Eq.~\eqref{eq:quadratic}. In this case, to decide whether the QCRB is attainable or not we simply check whether the quadratic equations are all satisfied for this $x$ for a sufficiently large kernel size (see Appendix~\ref{sec:finite-support}).

\section{Arbitrary outcomes reduce to finite outcomes}\label{sec:finite-support}
In this appendix we prove that the restriction to finite outcome POVMs (as used in Definition~\ref{def:target} and Theorem~\ref{theorem:finitecriterion}) is sufficient. Specifically we show that any model with only finite outcome POVMs can achieve the same mean, first order derivative and covariance. We start from a general arbitrary outcome POVM. Let $\Pi$ be a POVM on a measurable space $(\Omega,\Sigma)$ and let $e:\Omega\to\R^p$ be measurable.  Write
\begin{equation}
     p_0(F)=\Tr[\rho_0\Pi(F)]\;.
\end{equation}
The pair $(\Pi,e)$ is locally unbiased at $0$ if
\begin{equation}
  \int e_i\,dp_0=0,
  \qquad
  \int e_i(x)\,\Tr[\rho_j\Pi(dx)]=\delta_{ij}.
  \label{eq:localunbiased}
\end{equation}
Its covariance matrix is
\begin{equation}
  V=\int e(x)e(x)^T\,dp_0(x).
  \label{eq:covariance}
\end{equation}
Finite covariance means $\int e_i^2\,dp_0<\infty$ for every $i$.

\begin{proposition}[Finite-outcome POVMs are sufficient]
\label{prop:finite-support}
    Let $(\rho_0,\rho_1,\ldots,\rho_p)$ be a finite-dimensional model admitting Hermitian SLDs.  For any locally unbiased estimator of finite covariance on an arbitrary measurable outcome space, there is a finite-outcome estimator with exactly the same mean, derivative matrix, and covariance.  If $\dim\mathcal H=D$, at most $D^2+\frac{3p(p+1)}2$
outcomes are needed.
\end{proposition}

The discretization will be obtained using the classical Richter--Tchakaloff theorem \cite{Richter1957,BayerTeichmann}: if
$\mu$ is a positive measure, $E$ is a real vector space of dimension $n$, and
$\Phi:\Omega\to E$ is integrable, then, outside any prescribed $\mu$-null set,
there are points $x_1,\ldots,x_s$ and weights $\lambda_a>0$ such that for $s\le n$
\begin{equation}
  \int \Phi\,d\mu=\sum_{a=1}^s\lambda_a\Phi(x_a)\;.
  \label{eq:richter-tchakaloff}
\end{equation}
(See \cite[Satz~4]{Richter1957}; the precise measurable-map
form used here is \cite[Corollary~2]{BayerTeichmann}.)

\begin{proof}[Proof of Proposition~\ref{prop:finite-support}]
Let $(\Pi,e)$ be a locally unbiased estimator with finite covariance $V$.
Let $\mu(F)=\Tr [\Pi(F)]$, so that $\mu(\Omega)=D$. If $\mu(F)=0$, $\Tr [\Pi(F)]=0$, and $\Pi(F)\succeq0$ implies $\Pi(F)=0$. Hence every matrix element of the operator-valued measure $\Pi$ is absolutely continuous with respect to $\mu$. Applying the Radon--Nikodym theorem to each matrix element then gives a measurable $R:\Omega\to\Herm_D$ such
that, outside a $\mu$-null set,
\begin{equation}
  \Pi(F)=\int_F R(x)\,d\mu(x),
  \qquad R(x)\succeq0,
  \qquad \Tr [R(x)]=1,
  \qquad \int R(x)\,d\mu(x)=I.
  \label{eq:povmdensity}
\end{equation}
Going forward, we shall suppress the dependence on $x$ unless it is informative to include it. Thus \(R(x)=d\Pi/d\mu(x)\) is simply the trace-one density of the POVM; for a discrete POVM it corresponds to \(R_x=\Pi_x/\Tr[\Pi_x]\). Since \(R(x)\succeq0\) and \(\Tr [R(x)]=1\), \(R\) is bounded (in any matrix norm), and hence integrable because \(\mu(\Omega)=D<\infty\).

Fix SLDs $L_j$ and set
\begin{equation}
      a=\Tr[\rho_0R],
  \qquad c_j=\Tr[\rho_jR],
  \qquad b_j=\Tr[L_j\rho_0L_jR].
\end{equation}
Thus
\begin{equation}
    dp_0(x)=a(x)d\mu(x)\;,\qquad\Tr[\rho_j\Pi(dx)]=c_j(x)d\mu(x)
\end{equation}

The SLD equation and Hilbert--Schmidt Cauchy--Schwarz give, pointwise,
\begin{equation}
  |c_j|
  =\bigl|\operatorname{Re}\Tr[\rho_0L_jR]\bigr|\le\bigl|\Tr[\rho_0L_jR]\bigr|\le\sqrt{\Tr[\rho_0R]\Tr[L_j\rho_0L_jR]}
  = \sqrt{a\,b_j}\;,
  \label{eq:derivativebound}
\end{equation}
where we have applied Cauchy--Schwarz to $\rho_0^{1/2}R^{1/2}$ and $\rho_0^{1/2}L_jR^{1/2}$. Moreover,
\begin{equation}
      \int a\,d\mu=1,
  \qquad
  \int e_i^2a\,d\mu=V_{ii},
  \qquad
  \int b_j\,d\mu=\Tr[L_j\rho_0L_j]=J_{jj}\;.
\end{equation}
Hence $e_i\sqrt a$, $\sqrt a$, and $\sqrt{b_j}$ lie in $L^2(\mu)$ (the space of measurable functions whose square is integrable with respect to the measure $\mu$).
Equation~\eqref{eq:derivativebound} and Cauchy--Schwarz therefore show that
$e_i a$, $e_i c_j$, and $e_ie_k a$ are integrable:
\begin{equation}
    \begin{split}
        \int|e_ia|d\mu= \int|e_i|ad\mu\le\left(\int|e_i|^2ad\mu\right)^{1/2}\left(\int ad\mu\right)^{1/2}=\left(\int e_i^2ad\mu\right)^{1/2}=\sqrt{V_{ii}}<\infty\\
        \int|e_ic_j|d\mu\le\int|e_i|\sqrt{ab_j}d\mu\leq\left(\int|e_i|^2ad\mu\right)^{1/2}\left(\int b_jd\mu\right)^{1/2}=\sqrt{V_{ii}J_{jj}}<\infty\\
        \int|e_ie_k|ad\mu\le\left(\int|e_i|^2ad\mu\right)^{1/2}\left(\int|e_k|^2ad\mu\right)^{1/2}=\sqrt{V_{ii}V_{kk}}<\infty
    \end{split}
\end{equation}

Now apply \eqref{eq:richter-tchakaloff} to
\begin{equation}
  \Phi(x)=
  \left(
    R(x),
    (e_i(x)a(x))_i,
    (e_i(x)c_j(x))_{i,j},
    (e_i(x)e_j(x)a(x))_{1\le i\le j\le p}
  \right),
  \label{eq:momentmap}
\end{equation}
viewed in the real vector space
$\Herm_D\oplus\R^p\oplus\R^{p\times p}\oplus\Sym_p(\R)$ of dimension
\begin{equation}
  n=D^2+p+p^2+\frac{p(p+1)}2
   =D^2+\frac{3p(p+1)}2\;.  
\end{equation}  
We can think of $\Phi(x)$ as the vector of quantities whose integral value we wish to conserve. Apply Eq.~\eqref{eq:richter-tchakaloff} on the full-\(\mu\)-measure set on which \eqref{eq:povmdensity} and the pointwise bounds above hold. Define the finite estimator
\begin{equation}
  \Pi_r=\lambda_rR(x_r),
  \qquad
  \widehat e(r)=e(x_r).
  \label{eq:finitepovm}
\end{equation}
The $R$-coordinate gives $\sum_r\Pi_r=I$, and the remaining coordinates give
\begin{align*}
  \sum_r\widehat e_i(r)\Tr[\rho_0\Pi_r]
    &=\int e_i\,dp_0,\\
  \sum_r\widehat e_i(r)\Tr[\rho_j\Pi_r]
    &=\int e_i(x)\Tr[\rho_j\Pi(dx)],\\
  \sum_r\widehat e_i(r)\widehat e_j(r)\Tr[\rho_0\Pi_r]
    &=V_{ij}.
\end{align*}
Thus the finite estimator preserves the mean, derivative matrix, and covariance
exactly, with at most $n$ outcomes.
\end{proof}

\section{Deciding the exact attainability of the Nagaoka--Hayashi and Holevo Cram{\'{e}}r-Rao bounds is \np-hard}
\label{apen:NHCRHCRB}
Finally, in this appendix we prove that deciding the exact attainability of other important Cram{\'{e}}r-Rao bounds is \np-hard, extending the relevance of our results. The QCRB, Holevo Cram{\'{e}}r-Rao bound $\mathcal{C}_\text{H}$~\cite{holevo1973statistical,holevo2011probabilistic}, Nagaoka--Hayashi Cram{\'{e}}r-Rao bound $\mathcal{C}_\text{NH}$~\cite{nagaoka2005generalization,nagaoka2005new,hayashi1999,conlon2021efficient}, and most informative bound $\mathcal{C}_\text{MI}$ satisfy the following relation~\cite{conlon2022gap}
\begin{equation}
    \mathcal{C}_\text{MI}\geq\mathcal{C}_\text{NH}\geq \mathcal{C}_\text{H}\geq\mathcal{C}_\text{Q}\;.
\end{equation}
We now extend the \np-hardness of deciding whether there exists a $V$ such that $V=\mathcal{J}_\text{Q}^{-1}$ to the NHCRB and Holevo Cram{\'{e}}r-Rao bounds. Specifically we examine the hardness of deciding whether there exists a $V$ such that $\Tr[V]=\mathcal{C}_\text{NH}$ and $\Tr[V]=\mathcal{C}_\text{H}$. Given a multiparameter estimation problem defined by $\rho_\theta$, we define $\mathbb{S}_\theta= {1}_p\otimes \rho_\theta$. We can now state the NHCRB as:  \begin{equation}
\Tr[V]\geq\mathcal{C}_\text{NH}= \min_{\mathbb{L},\,X}\left\{
                          \mathbb{\Tr}[\mathbb{S}_\theta
                          \mathbb{L}]\,\big|\, \mathbb{L}_{jk}=\mathbb{L}_{kj}\, \mathrm{
       Hermitian, }\,\mathbb{L}\geq {X} X^\intercal,\, X_j\,\mathrm{Hermitian\,satisfying\,\eqref{eq:Xunbiased} }
    \right\} \;,\label{eq:gnb}
\end{equation}
where $\mathbb{L}$
is a $p$-by-$p$ matrix of Hermitian operators $\mathbb{L}_{jk}$, and the unbiased conditions are
\begin{align}
  \label{eq:Xunbiased}
  \Tr{\rho_0 X_j}= 0 \qquad\text{and}\qquad \frac{\partial}{\partial \theta_j} \Tr{\rho_\theta X_k}|_{\theta=0}=\delta_{jk}\;.
\end{align}
We now prove that for the hard model defined in Theorem~\ref{th:quantum-encoding}, $\mathcal{C}_\text{NH}=\mathcal{C}_\text{H}=\mathcal{C}_\text{Q}$, which immediately proves the desired result.

\begin{lemma}
    For the hard model defined in Theorem~\ref{th:quantum-encoding} with $\Tr[H]=1$, $\mathcal{C}_\text{NH}=\mathcal{C}_\text{H}=\mathcal{C}_\text{Q}$.
\end{lemma}
\begin{proof}
The SLD operators for the hard model are given by 
\begin{equation}
  L_i^{(0)}=
  \begin{pmatrix}
    T_i\oplus0_{\mathcal S}&\binom{0}{B_i}\\[2mm]
    (\,0\ \ B_i^T\,)&0_\mathcal{N}
  \end{pmatrix}\;.
\end{equation}
Recalling that $H_{ij}=B_iB_j^T$ we have that
\begin{equation}
    L_i^{(0)}L_j^{(0)}=T_iT_j\oplus H_{ij}\oplus B_i^TB_j\;.
\end{equation}
Using this we now construct a feasible solution to the NHCRB that gives $\mathcal{C}_\text{Q}=\mathcal{C}_\text{NH}$. Start from the matrix 
\begin{equation}
  Z_{ij}=T_iT_j\oplus H_{ij}\oplus\delta_{ij} I_q\;.
\end{equation}
We define $L=[L_1^{(0)},...,L_p^{(0)}]^T$ We first prove $Z\succeq LL^T$ before relating this to the NHCRB. The kernel space is the only space on which $Z-LL^T$ is non-zero. On the kernel space $Z-LL^T$ is given by
\begin{equation}
    I_{pq}-\tilde{B}\tilde{B}^T\;,
\end{equation}
where $\tilde{B}=[B_1^T,...B_p^T]^T$
 Observe that $R=\tilde{B}\tilde{B}^T\succeq0$ and $\Tr[R]=\sum_i\Tr[B_i^TB_i]=\Tr[H]=1$. Therefore, the eigenvalues of $R$ are at most 1 and so $I_{pq}\succeq R$. This proves $Z\succeq LL^T$.

 Now choose $X_i=\sum_j(\mathcal{J}_\text{Q}^{-1})_{ij}L_j^{(0)}$ and $\mathbb{L}_{ij}=\sum_{ab}(\mathcal{J}_\text{Q}^{-1})_{ia}(\mathcal{J}_\text{Q}^{-1})_{jb}Z_{ab}$. This solution $\mathbb{L}$ satisfies $\mathbb{L}_{ij}=\mathbb{L}_{ji}$ and each $\mathbb{L}_{ij}$ is Hermitian. The vector $X=(X_1,...,X_p)$ satisfies the unbiased conditions. Finally, $\mathbb{L}\succeq XX^T$ follows from $Z\succeq LL^T$. 

Using Eq.~\eqref{eq:qfiTH} we see that $\Tr[\rho Z_{ab}]=\mathcal{J}_{\text{Q},ab}$. The NHCRB can then be computed as
 \begin{equation}
\Tr[\mathbb{S}_\theta\mathbb{L}]=\sum_{i=1}^{p}\Tr[\rho \mathbb{L}_{ii}]=\sum_{i,a,b}(\mathcal{J}_\text{Q}^{-1})_{ia}(\mathcal{J}_\text{Q}^{-1})_{ib}\mathcal{J}_{\text{Q},ab}=\Tr[\mathcal{J}_\text{Q}^{-1}\mathcal{J}_\text{Q}(\mathcal{J}_\text{Q}^{-1})^T]=\mathcal{C}_\text{Q}\;.
 \end{equation}
 Therefore $\mathcal{C}_\text{NH}=\mathcal{C}_\text{Q}$, which immediately implies $\mathcal{C}_\text{NH}=\mathcal{C}_\text{H}=\mathcal{C}_\text{Q}$.
\end{proof}

Therefore, for the class of hard problems considered in this paper if one had access to an oracle which decided whether there exists a $V$ such that $\Tr[V]=\mathcal{C}_\text{NH}$, we could use this oracle to decide whether $V=\mathcal{J}_\text{Q}^{-1}$. Therefore, the problem of deciding whether $\Tr[V]=\mathcal{C}_\text{NH}$ or $\Tr[V]=\mathcal{C}_\text{H}$ is also \np-hard.

\end{document}